\documentclass[11pt,a4paper]{article}

\usepackage{geometry} 
\usepackage{setspace}
\usepackage[utf8]{inputenc}
\usepackage[T1]{fontenc}
\usepackage{lmodern}
\usepackage{microtype}

\usepackage{amsmath,amssymb,amsfonts,amsthm}
\usepackage{mathtools}
\DeclareMathOperator{\sech}{sech}
\DeclareMathOperator{\arccot}{arccot}
\usepackage{graphicx}
\usepackage{subcaption}
\usepackage{graphicx}
\usepackage{caption}
\usepackage{subcaption}
\usepackage{booktabs}
\usepackage{tabularx}
\usepackage{array}
\usepackage{multirow}
\usepackage{multicol}
\usepackage{float}
\usepackage{algorithm}
\usepackage{algpseudocode}
\usepackage{booktabs}

\usepackage{xcolor}
\usepackage{pifont}

\newcommand{\cmark}{\textcolor{green!60!black}{\ding{51}}}
\newcommand{\xmark}{\textcolor{red!75!black}{\ding{55}}}
\newcolumntype{Y}{>{\centering\arraybackslash}X}

\usepackage{algorithm}
\usepackage{algpseudocode}
\usepackage{listings}

\definecolor{mild}{rgb}{1,0.98,0.8}
\usepackage{comment}
\usepackage{todonotes}

\usepackage{enumitem}
\usepackage{url}
\usepackage{natbib}

\usepackage{hyperref}

\hypersetup{
  colorlinks=true,
  linkcolor=blue,
  citecolor=black,
  urlcolor=blue
}

\newtheorem{theorem}{Theorem}[section]
\newtheorem{lemma}[theorem]{Lemma}
\newtheorem{corollary}[theorem]{Corollary}
\newtheorem{proposition}[theorem]{Proposition}

\theoremstyle{definition}

\theoremstyle{remark}

\numberwithin{equation}{section}
\numberwithin{figure}{section}
\numberwithin{table}{section}

\usepackage[section]{placeins}
\usepackage{flafter}

\title{\bf Structure-Informed Neural Operators for Long-Time Prediction of Parametric Hamiltonian PDEs}
\author{
Victory C. Obieke\thanks{Corresponding author: Victory C. Obieke,  Oregon State University. Email: \texttt{obiekev@oregonstate.edu}}\\
Oregon State University\\
\texttt{obiekev@oregonstate.edu}
\and
Christopher Chukwuemeka\\
University of Houston\\
\texttt{ccchukw2@cougarnet.uh.edu}
\and
Emmanuel E. Oguadimma\\
Oregon State University\\
\texttt{oguadime@oregonstate.edu}
}
\begin{document}
\date{}
\maketitle

\begin{abstract}
Hamiltonian partial differential equations (PDEs) often exhibit long-time dynamics governed by conserved quantities such as mass, momentum, and Hamiltonian energy. Standard Fourier neural operators (FNOs) provide efficient data-driven approximations of solution operators, but may not preserve these invariants during autoregressive rollout, and can develop drift in conserved quantities, phase error, and loss of qualitative accuracy. We propose an energy-projection Fourier neural operator (EP-FNO), a structure-informed operator learning architecture that combines a residual FNO time-stepping update with an invariant projection for long-time prediction of parametric Hamiltonian PDEs. We also provide a theoretical analysis showing that EP-FNO can approximate operators
associated with PDEs efficiently, we also suggest a stability estimate.
 We evaluate the approach on the Zakharov--Kuznetsov, Kadomtsev--Petviashvili, and sine--Gordon equations. Numerical experiments show that the projected model improves long-time stability, and gives more accurate propagation of soliton and coherent wave structures compared with a standard FNO baseline. Our results demonstrate that invariant projection improves the reliability of learned surrogates for long-time Hamiltonian PDE simulation.
\end{abstract}

\section{Introduction}

Hamiltonian partial differential equations (PDEs) arise in computational science, including nonlinear wave propagation \cite{whitham1974linear,drazin1989solitons,sulem1999nonlinear}, plasma physics \cite{morrison2005hamiltonian}, fluid dynamics and magnetohydrodynamics \cite{morrison1980noncanonical,morrison1998hamiltonian}, and nonlinear optical media \cite{boyd2020nonlinear,agrawal2019nonlinear}.
These systems often support coherent wave structures, such as solitons and localized waves, whose long-time dynamics are governed by conserved quantities including mass, momentum, and Hamiltonian energy. Accurately predicting these dynamics over long time intervals remains challenging, since small numerical or model errors can accumulate during time integration, producing phase errors, artificial damping, or drift in conserved quantities.

Numerical methods such as finite-difference and FDTD schemes
\cite{Yee1966, bokil2018highFDTD, gibson2026analysis},
finite element methods \cite{banks2009finiteElement,monk2003finiteElement},
spectral and pseudospectral methods \cite{trefethen2000spectral,fornberg1996pseudospectral},
and multisymplectic discretizations
\cite{bridges2001multisymplecticSpectral,chen2001multi,chen2011multi},
have long been used to simulate Hamiltonian and dispersive wave systems. Beyond formal discretization accuracy, the effectiveness of these methods often depends on how well they reflect the structure of the underlying differential equation, including symmetries, invariant transformations, and conserved quantities \cite{oguadimma2026foundational}.
Structure-preserving schemes are therefore especially important because they are designed to respect the geometric properties of the continuous model
\cite{obieke2026energy,chen2011multi}.
However, high-fidelity simulations can still be computationally expensive, especially when solutions must be generated repeatedly for many initial conditions or parameter values. This has led to growing interest in machine-learning methods that approximate the evolution of nonlinear systems from data. In particular, there has been significant progress in incorporating Hamiltonian and geometric structure into learned dynamical models. Hamiltonian neural networks learn an energy function and use Hamilton's equations to evolve the state \cite{greydanus2019hamiltonian}. SympNets and other related architectures instead learn maps that are constrained to preserve the symplectic structure of Hamiltonian dynamics \cite{jin2020sympnets,burby2021normal}. Other approaches incorporate physical constraints or conservation laws through penalty terms, variational formulations, or physics-informed residual losses \cite{raissi2019physics,kissas2020machine,chen2022kam,zhang2022vanishing}.

In contrast to learning maps between finite-dimensional vectors \cite{li2020fourier,kovachki2023neural}, neural operators provide a data-driven framework for learning solution operators between infinite-dimensional function spaces. Among existing operator-learning methods, the Fourier neural operator (FNO) has become a standard architecture for learning PDE solution maps on regular grids, since it parameterizes integral kernels in Fourier space and uses spectral convolutions to capture global interactions efficiently \cite{li2020fourier}. Several extensions of the FNO framework have been proposed, including factorized Fourier neural operators \cite{tran2021factorized}, residual factorized Fourier neural operators \cite{dauner2024residual}, and geometry-aware Fourier neural operators \cite{li2023fourier} for more complex geometries. However, explicit enforcement of the invariant structure of the underlying PDE during autoregressive rollouts remains underexplored. As the predicted solution is repeatedly fed back into the model, small one-step errors can accumulate into long-time drift in conserved quantities as in fig ~\ref{fig:epfnointro}. 
\begin{figure}[htbp]
    \centering
    \includegraphics[width=0.9\linewidth]{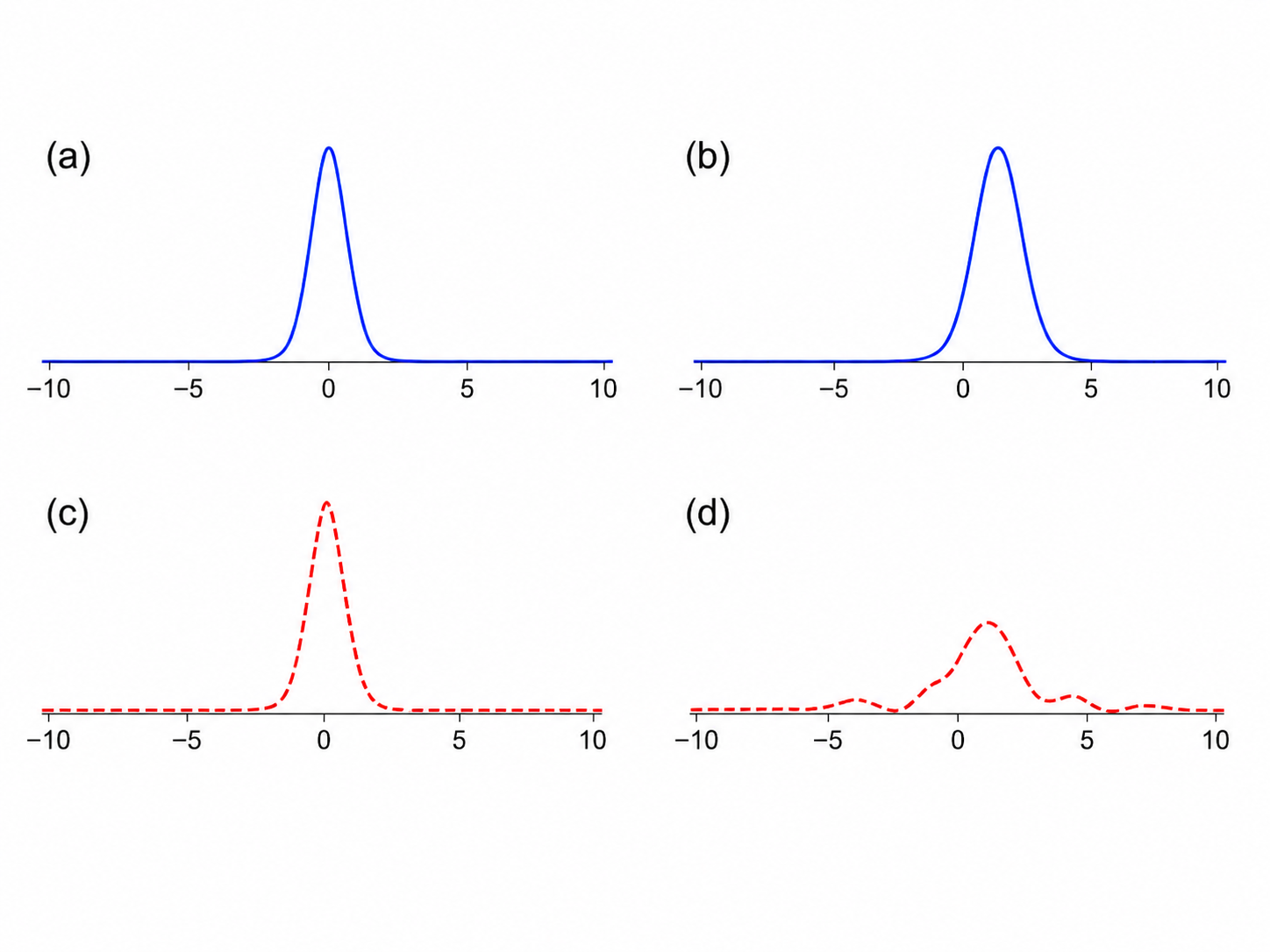}
    \caption{\textbf{Conceptual illustration of structure preservation in long-time soliton prediction.} (a) Initial soliton supplied to the structure preserving neural operator at (t=0). (b) At (t=2), the a structure preserving neural operator prediction retains a localized, coherent soliton profile. (c) The same initial soliton supplied to the standard FNO. (d) At (t=2), the FNO prediction exhibits amplitude loss, spreading, and oscillatory artifacts, illustrating degradation of the soliton structure during autoregressive rollout.}
    \label{fig:epfnointro}
\end{figure}

In this paper, we consider learning long-time solution operators for a family of Hamiltonian PDEs. We propose the energy-projection Fourier neural operator (EP-FNO), a structure-informed FNO architecture that combines a residual FNO time-stepping update with an invariant projection for long-time prediction of parametric Hamiltonian PDEs. The FNO algorithm learns the local solution increment from a history window of previous states, while the projection module corrects the prediction toward a discrete invariant
manifold. This projection step parallels residual-minimizing projection methods in reduced-order modeling, where a discrete correction is used to improve the reliability of an approximate time update \cite{carlberg2013gnat}.

The remainder of this paper is organized as follows. 
In Section~\ref{sec:problem_setting}, we formulate the long-time operator-learning problem for Hamiltonian PDEs and introduce the discrete invariant manifold used in the projection step. Section~\ref{sec:relwork} reviews related work on Fourier neural operators, Hamiltonian and symplectic learning, structure-preserving neural operators, physics-informed methods, and geometric numerical methods. 
In Section~\ref{sec:epfno}, we present the proposed EP-FNO architecture, including the neural-operator algorithm, Fourier operator layers, residual time-step update, invariant projection block, and stability estimates. Section~\ref{sec:numerical_experiments} reports numerical experiments for the Zakharov--Kuznetsov, Kadomtsev--Petviashvili, and sine--Gordon equations, together with ablation studies and cost--accuracy comparisons. Finally, Section~\ref{sec:discussion} summarizes the main findings and discusses possible extensions of the projection framework.

\section{Problem Setting}
\label{sec:problem_setting}
In this section, we consider the problem of learning solution operators for Hamiltonian evolution equations. Let \(\Omega\subset\mathbb{R}^d\) be a spatial domain and let $T>0$. We denote by \(U:\Omega\times(0,T]\to\mathbb{R}^m\) the solution of the system.  We assume that the governing equation can be written in first-order form as
\begin{equation}
\label{eq:first_order_pde}
\begin{aligned}
\partial_t U(x,t) &= \mathcal F[U](x,t),
&& (x,t)\in \Omega\times(0,T],\\
U(x,0) &= U_0(x),
&& x\in\Omega.
\end{aligned}
\end{equation}
where \(\mathcal F\) denotes the spatial differential operator associated with the system. We further assume that the flow generated by \eqref{eq:first_order_pde} admits \(q\) conserved functionals $\mathbf C(U)=\big(C_1(U),\ldots,C_q(U)\big)^\top$, meaning that, for each solution \(U(t)\), one has 
\begin{align}\label{derivative}
    \frac{d}{dt}C_i(U(t))=0, \qquad i=1,\ldots,q.
\end{align}
Thus, \(C_i(U(t))=C_i(U_0)\) along the exact evolution. These conserved
quantities may include mass, momentum, or Hamiltonian energy.

After using any desired choice of numerical method for the spatial discretizations of \eqref{eq:first_order_pde}, we let \(\mathbb{V}_h\) denote the discrete state space and let \(U_h^n\in \mathbb{V}_h\) be the approximation of \(U(\cdot,t_n)\), where \(t_n=n\Delta t,\, n\in \mathbb{N} \). We denote the discrete conserved quantities by $\mathbf C_h(U_h)=\big(C_{1,h}(U_h),\ldots,C_{q,h}(U_h)\big)^\top$. Given the input history $\mathbf U_h^n=\big(U_h^{n-T_{\rm in}+1},\ldots,U_h^n\big)$ where \(T_{\rm in}\) is the length of the history window, the solution operator is
\begin{align}
    \mathcal G^\dagger:\mathbf U_h^n \longmapsto U_h^{n+1}.
\end{align}
We approximate \(\mathcal G^\dagger\) by a learned neural operator. In the standard FNO model, the next state is predicted by the map $\widehat U_h^{n+1}=\mathcal G_{\theta}(\mathbf U_h^n),$  $\theta\in \Theta,$ for some finite-dimensional parameter space $\Theta$. In general,
\(\widehat U_h^{n+1}\) need not satisfy the discrete conservation laws. To enforce the invariant structure, we define the target invariant vector \(\mathbf c^\star=\mathbf C_h(U_h^n)\), and the corresponding discrete invariant manifold
\begin{align}
    \mathcal M=\big\{W_h\in \mathbb{V}_h:\mathbf C_h(W_h)=\mathbf c^\star\big\}.
\end{align}
The EP-FNO update is defined by
\begin{align}
U_h^{n+1}
=
\mathcal G_{\theta}^{\rm EP}(\mathbf U_h^n)
=
\Pi_{\mathbf c^\star}
\left(\mathcal G_{\theta}(\mathbf U_h^n)\right),
\end{align}
where \(\Pi_{\mathbf c^\star}\) denotes the projection map defined in Section~\ref{subsec:invariant_projection_block}. This projection corrects the FNO prediction so that it lies on, or closer to, the discrete invariant manifold \(\mathcal M\). Thus, the FNO algorithm provides the data-driven prediction, while the projection step reduces drift in the discrete invariants during autoregressive rollout.
\begin{figure}[H]
    \centering
    \includegraphics[width=0.9\linewidth]{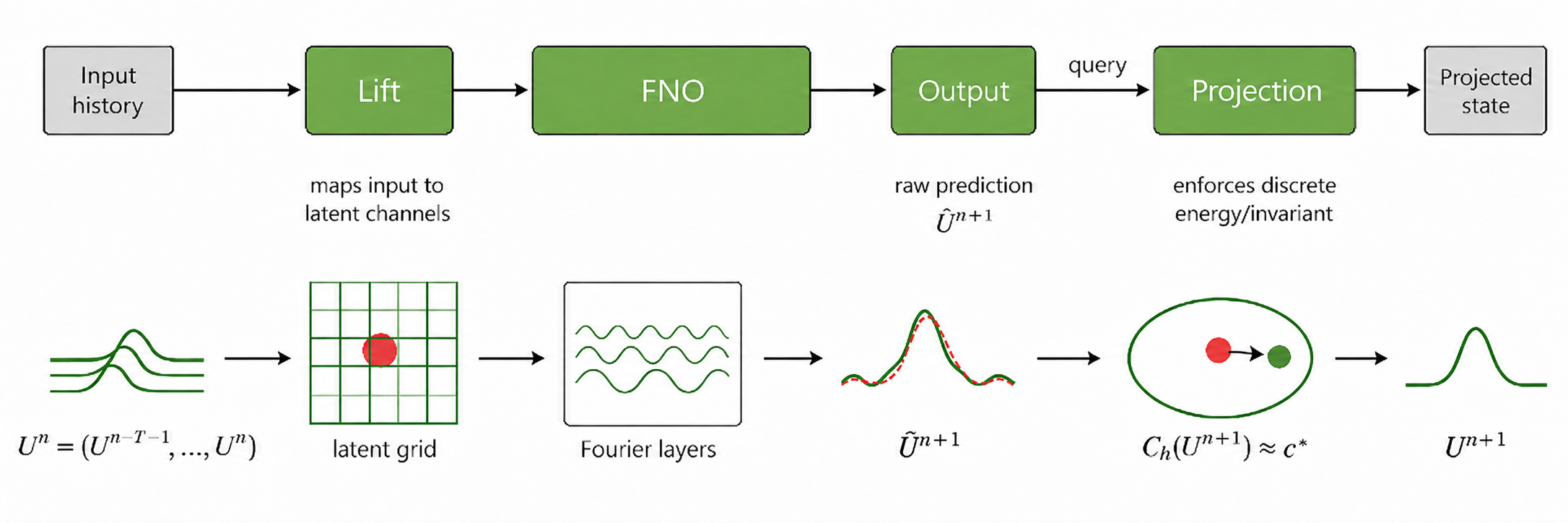}
\caption{Architecture of the proposed energy-projection Fourier neural operator (EP-FNO). The input history is lifted into a latent representation and evolved by an FNO algorithm to produce a prediction. The projection module then corrects this prediction so that the final state satisfies the discrete invariant.}
    \label{fig:epfno_architecture}
\end{figure}

The overall EP-FNO architecture is shown in Figure~\ref{fig:epfno_architecture}. The model
has two main components: (i) a residual FNO backbone that predicts the
next-step increment of the PDE state, and (ii) an invariant projection block that corrects the provisional prediction so that it remains close to the prescribed Hamiltonian energy level set. The computational cost and structure-preserving properties of FNO and EP-FNO are
summarized in Table~\ref{tab:fno_epfno_complexity_structure}, while the corresponding
stability properties are summarized in Table~\ref{tab:stability_fno_epfno}. 

\begin{table}[htbp]
\centering
\caption{Computational complexity and structure-preserving comparison of FNO and EP-FNO models.
Here \(N=N_xN_y\) is the number of spatial grid points, \(T\) is the rollout length, and rs denotes residual time stepper.}
\label{tab:fno_epfno_complexity_structure}
\begin{tabular}{lccc}
\hline
\textbf{Model} 
& \textbf{\(T\)-step cost} 
& \textbf{Invariant-Projection} \\
\hline
FNO 
& \(O(TN\log N)\) 
& \xmark  \\

EP-FNO without rs
& \(O(TN\log N)+O(TN)\) 
& \cmark \\

EP-FNO 
& \(O(TN\log N)+O(TN)\) 
& \cmark \\
\hline
\end{tabular}

\vspace{0.5em}
\begin{minipage}{0.92\textwidth}
\end{minipage}
\end{table}
\begin{table}[htbp]
\centering
\caption{Stability comparison of FNO and EP-FNO. Here rs denote residual time-stepper \eqref{subsec:residual_time_update}.}
\label{tab:stability_fno_epfno}
\begin{tabular}{lccc}
\hline
\textbf{Model}
& \textbf{Energy boundedness}
& \textbf{Long-time stability}
& \textbf{Residual stability} \\
\hline

FNO
& \xmark
& \xmark
& \xmark \\

EP-FNO without rs
& \cmark
& \cmark
& \xmark \\

EP-FNO
& \cmark
& \cmark
& \cmark \\

\hline
\end{tabular}

\vspace{0.5em}
\begin{minipage}{0.88\textwidth}
\end{minipage}
\end{table}

\section{Related Works}\label{sec:relwork}
\paragraph{Fourier neural operators and neural operator learning.}
The Fourier neural operator (FNO) \cite{li2020fourier,kovachki2023neural} is a standard approach for learning solution operators of PDEs on regular grids. By parameterizing the integral kernel in Fourier space, FNOs achieve resolution independence and excellent performance on a variety of problems including Navier--Stokes, Darcy flow, and Burgers' equation. Several extensions have been proposed, such as the factorized FNO \cite{tran2021factorized} and the geo-FNO \cite{li2023fourier} for irregular geometries. However, these architectures do not explicitly enforce conservation laws or Hamiltonian structure, which can lead to qualitative drift in long-time autoregressive rollout.

\paragraph{Hamiltonian and symplectic neural networks.}
For ordinary differential equations, the Hamiltonian neural network (HNN) \cite{greydanus2019hamiltonian} learns the Hamiltonian function and uses it to generate symplectic dynamics. SympNets \cite{jin2020sympnets} and their variants \cite{burby2021normal} directly parameterize symplectic maps. These methods guarantee exact conservation of the learned Hamiltonian (up to discretization) and have been extended to time-dependent Hamiltonians \cite{zhong2024symp} and port-Hamiltonian systems \cite{dipersio2025porthamiltonian}. 


\paragraph{Structure-preserving neural operators.}
Recent efforts aim to embed geometric structure into neural operators. The structure-preserving neural differential operator with embedded Hamiltonian constraints \cite{najera2023structure} introduces a symplectic latent representation for Hamiltonian systems, while the Lagrangian neural operator (LNO) \cite{cranmer2020lagrangian} learns a Lagrangian density from data. Another line of work uses penalty terms in the loss to approximately conserve energy or momentum \cite{kissas2020machine}, but this approaches does not guarantee exact preservation. The projection method proposed in EP-FNO is conceptually closer to the orthogonal projection used in projection-based model reduction \cite{carlberg2013gnat}. 

\paragraph{Physics-informed learning for Hamiltonian PDEs.}
Physics-informed neural networks (PINNs) \cite{raissi2019physics, obieke2025structure} can incorporate Hamiltonian structure by minimizing the residual of the Hamilton--Jacobi or the variational form \cite{chen2022kam, zhang2022vanishing}.  While PINNs excel at inverse problems and irregular domains, they are typically slower than neural operators for repeated forward simulations and do not directly produce a discrete-time map. EP-FNO is designed for fast autoregressive rollout and inherits the efficiency of spectral convolutions.


\paragraph{Projection methods for physical constraints.}
A closely related direction enforces physical consistency by correcting learned
outputs after prediction. General output-projection frameworks formulate the
correction as a constrained post-processing problem, where the prediction is
moved onto a manifold defined by prescribed physical laws
\cite{valente2025physics}. In operator learning, hard-constraint methods project
surrogate outputs onto spaces of functions satisfying prescribed differential
constraints, including efficient transformed-space projections for linear
constraints such as incompressibility \cite{duruisseaux2024hard}. Other recent
approaches introduce adaptive correction modules that modify neural-operator
outputs to enforce linear or quadratic conserved quantities \cite{liu2025adaptive}. Projection layers have also been used in conservative
PINNs and DeepONets by mapping candidate solutions of dynamical systems with first integrals onto the corresponding invariant manifold \cite{cardosobihlo2025exactly}. These works establish projection and correction
as important tools for scientific machine learning. EP-FNO differs by coupling a
residual FNO time-stepper with an explicit, sample-dependent damped projection
onto problem-dependent discrete Hamiltonian and mass manifolds during long
autoregressive rollout. This allows the correction to target nonlinear
Hamiltonian invariants in parametric two-dimensional wave PDEs, which makes the
specific EP-FNO formulation proposed here novel.


\section{Structure-Informed Fourier Neural Operator}
\label{sec:epfno}

We propose a structure-informed FNO for learning solution
operators associated with Hamiltonian evolution equations. The method combines
the FNO algorithm with a residual time-step update and a
structure-preserving projection module. The FNO backbone learns the local
increment of the PDE state from a history window of previous states, producing
a provisional prediction \(\widetilde U^{n+1}=U^n+\mathcal G_\theta(\mathbf U^n)\).
This provisional state is then passed to an invariant projection module, which
corrects the prediction so that it is consistent with prescribed discrete
quantities such as mass or Hamiltonian energy. In this way, the architecture
retains the approximation power of FNO while incorporating the geometric structure of the underlying PDE during autoregressive rollout.

\subsection{Neural Operator}
\label{subsec:neural_operator}

Let \(\Omega\subset\mathbb{R}^d\) denote the spatial domain. A neural operator maps an input function \(a:\Omega\to \mathbb{R}^{d_a}\) to an output function \(U:\Omega\to \mathbb{R}^{d_u}\). 
We write the operator as a composition of pointwise lifting $P(a)$, integral operator layers, and pointwise output projection $Q(.)$,
\begin{align}\label{op}
    \Psi_\theta=Q\circ K_L\circ \cdots \circ K_1\circ P.
\end{align}
The lifting map \(P\) sends the input function into a higher-dimensional latent space, \(v_0(x)=P(a)(x)\), \(x\in \Omega\), while the final map \(Q\) projects the last latent representation back to the physical output space, \(\widehat U(x)=Q(v_L)(x)\). The intermediate layers \(K_\ell\), \(\ell=1,\ldots,L\), are operator layers of the form \(v_\ell(x)=\int_\Omega \kappa_\ell(x,y)v_{\ell-1}(y)\,dy\), where \(\kappa_\ell\) is a learnable kernel. Nonlinear activation functions are
applied between layers.

\subsection{Fourier Operator Block}
\label{subsec:fourier_operator_block}

We describe the Fourier operator block used in the FNO algorithm. From \eqref{op}, the lifted latent representation \(v_0\), evaluated on a regular grid, is passed through a sequence of Fourier layers. Let the latent functions be defined on the \(d\)-dimensional torus \(\mathbb T^d\). We define an integral operator with kernel
\(\kappa \in L^2(\mathbb T^d;\mathbb R^{n\times m})\) as the mapping
\(\mathcal C:L^2(\mathbb T^d;\mathbb R^m)\to L^2(\mathbb T^d;\mathbb R^n)\)
given by
\[
    \mathcal C(v)=\mathcal F^{-1}\left(\mathcal F(\kappa)\cdot \mathcal F(v)\right),
\]
with \(v\in L^2(\mathbb T^d;\mathbb R^m)\), where \(\mathcal F\) and \(\mathcal F^{-1}\) denote the Fourier transform and inverse Fourier transform, respectively. 
The Fourier layer is then defined by
\[
    K(v)(x)=\sigma\left(Wv(x)+\mathcal C(v)(x)\right), \quad x\in\mathbb T^d
\]
where \(W\in\mathbb R^{n\times m}\) is a learnable pointwise linear map and
\(\sigma\) is a nonlinear activation function.

Applying \(L\) such Fourier layers gives
\(v_L=K_L\circ K_{L-1}\circ \cdots \circ K_1(v_0)\).
The final latent representation is then mapped to the output space by a
pointwise projection operator \(Q\), giving the FNO prediction
\(\widehat U=\mathcal G_\theta(a)=Q(v_L)\).

\subsection{Residual Time-Step Update}
\label{subsec:residual_time_update}

In the autoregressive rollout, we train the FNO model to predict an increment rather than the full next state. Given the history window
\(a:=\mathbf U^n\), the provisional prediction can be written as
\[
    \widetilde U^{n+1} = R_\theta(\mathbf U^n) = U^n+\mathcal G_\theta(\mathbf U^n).
\]
Thus, the neural operator learns the local change in the PDE state over one
time step. This residual time-step formulation is useful for autoregressive
prediction because consecutive states are typically close for sufficiently
small \(\Delta t\). The provisional state \(\widetilde U^{n+1}\) is then passed to the invariant
projection module.

\subsection{Invariant Projection Block}
\label{subsec:invariant_projection_block}

The residual FNO update produces a provisional prediction
\(\widetilde U^{n+1}=U^n+\mathcal G_\theta(\mathbf U^n)\). Since this prediction is not guaranteed to preserve the discrete
invariants of the underlying Hamiltonian system, we apply a post-processing
projection step.

Let \(\mathbf C_h(U)\) denote the
vector of discrete invariants to be stabilized. At each time step, we define the target invariant value by
\(\mathbf c^\star=\mathbf C_h(U^n)\). The projected EP-FNO update is then
defined by 
\[
    U^{n+1}=\Pi_{\mathbf c^\star}(\widetilde U^{n+1}) 
\]
where \(\Pi_{\mathbf c^\star}\) is a projection map chosen so that \(\mathbf C_h(U^{n+1})\approx \mathbf c^\star\). The projection serves as a structure-preserving correction applied after the prediction made by the FNO algorithm. 

For a general Hamiltonian PDE, the low-dimensional correction map is written in terms of correction directions as
\begin{align}\label{correct}
     \Pi_{\boldsymbol\lambda}(V)
    =
    V
    +
    \sum_{j=1}^{q}
    \lambda_j D_j(V),
\end{align}
where \(D_j(V)\), \(j=1,\ldots,q\), are chosen correction directions. They tell the projection which direction to move the raw prediction V in order to reduce the invariant defect and
\(\boldsymbol\lambda=(\lambda_1,\ldots,\lambda_q)^\top\) is the vector of
correction parameters. 
\subsection{Implementation of the EP-FNO Algorithm}

In this section, we describe the implementation of EP-FNO. At each
time level, the residual FNO first produces a provisional next state from the
current history window. The prescribed discrete invariants are then computed
from the last accepted state, and a low-dimensional damped projection is applied
to the provisional prediction. This projection is recomputed for each sample and
each rollout step, so the FNO backbone remains unchanged while the correction
reduces drift in the prescribed invariants during long autoregressive prediction.

\begin{algorithm}[htpb]
\caption{EP-FNO with damped invariant projection algorithm}
\label{alg:epfno_projection}
\begin{algorithmic}[1]

\State \textbf{Input.} History window $ \mathbf U^n = \bigl( U^{n-T_{\rm in}+1},\ldots,U^n\bigr),$ residual FNO \(\mathcal G_\theta\), discrete invariant map $ \mathbf C_h(U)$ projection family \(\Pi_{\boldsymbol\lambda}\), identity correction
\(\boldsymbol\lambda_0\), damping parameter \(0<\eta\le 1\), and Newton
tolerance \(\mathrm{tol}\).

\State \textbf{Residual FNO prediction.}
Compute the provisional next state
\[
    \widetilde U^{n+1}
    =
    U^n+\mathcal G_\theta(\mathbf U^n).
\]

\State \textbf{Target invariant value.}
Compute the invariant vector from the last accepted state,
\[
    \mathbf c^\star
    =
    \mathbf C_h(U^n).
\]

\State \textbf{Invariant defect.}
For a provisional state \(V\), define (using \eqref{correct})
\[
    \mathbf F(\boldsymbol\lambda;V)
    =
    \mathbf C_h\!\left(
    \Pi_{\boldsymbol\lambda}(V)
    \right)
    -
    \mathbf c^\star .
\]

\State \textbf{Low-dimensional correction.}
Starting from \(\boldsymbol\lambda^{(0)}=\boldsymbol\lambda_0\), apply Newton
iterations to solve
\[
    \mathbf F(\boldsymbol\lambda_\ast;\widetilde U^{n+1})
    =
    \mathbf 0.
\]

\State \textbf{Damped correction.}
Set
\[
    \overline{\boldsymbol\lambda}
    =
    \boldsymbol\lambda_0
    +
    \eta
    \left(
    \boldsymbol\lambda_\ast-\boldsymbol\lambda_0
    \right).
\]

\State \textbf{Projected EP-FNO state.}
Define
\[
    U^{n+1}
    =
    \Pi_{\mathbf c^\star}(\widetilde U^{n+1})
    =
    \Pi_{\overline{\boldsymbol\lambda}}
    (\widetilde U^{n+1}).
\]

\State \textbf{Output.} Projected state \(U^{n+1}\).

\end{algorithmic}
\end{algorithm}

\

\begin{theorem}[{\cite[Theorem 5]{kovachki2021universal}}]\label{thm:universalapproximation}
Let \(s,s' \ge 0\). Let
$
\mathcal G : H^s(\mathbb T^d;\mathbb R^{d_a})
\to
H^{s'}(\mathbb T^d;\mathbb R^{d_u})
$
be a continuous operator. Let
\(K \subset H^s(\mathbb T^d;\mathbb R^{d_a})\)
be a compact subset. Then for any \(\epsilon > 0\), there exists a FNO
$
\mathcal N : H^s(\mathbb T^d;\mathbb R^{d_a})
\to
H^{s'}(\mathbb T^d;\mathbb R^{d_u}),
$
of the form (6), continuous as an operator \(H^s \to H^{s'}\), such that
\[
\sup_{a\in K}
\left\|
\mathcal G(a)-\mathcal N(a)
\right\|_{H^{s'}}
\le \epsilon .
\]
\end{theorem}

\begin{lemma}[\textbf{Residual FNO approximation}]
\label{lem:residual_fno_approximation}
Let \(s,s'\ge 0\), and let \(\mathbf U^n\) denote the history window whose
most recent state is \(U^n\). Let
$
\Phi_h:H^s(\mathbb T^d;\mathbb R^{d_a})
\to H^{s'}(\mathbb T^d;\mathbb R^{d_u})
$
be the reference one-step flow map on a compact set
\(K\subset H^s(\mathbb T^d;\mathbb R^{d_a})\), so that
$
    \Phi_h(\mathbf U^n)=U^{n+1}.
$
Define the residual operator by
$
    G(\mathbf U^n):=\Phi_h(\mathbf U^n)-U^n,
$
where \(U^n\) is the most recent state in the history window. Assume that
\(G\) is well-defined and continuous as a map
$
    G:K\to H^{s'}(\mathbb T^d;\mathbb R^{d_u}).
$
Then, for every \(\delta>0\), there exists an FNO
$
\mathcal G_\theta:H^s(\mathbb T^d;\mathbb R^{d_a})
\to H^{s'}(\mathbb T^d;\mathbb R^{d_u})
$
such that the residual update
$
    R_\theta(\mathbf U^n)=U^n+\mathcal G_\theta(\mathbf U^n)
$
satisfies
\[
    \sup_{\mathbf U^n\in K}
    \|R_\theta(\mathbf U^n)-\Phi_h(\mathbf U^n)\|_{H^{s'}}
    \le \delta .
\]
\end{lemma}

\begin{proof}
By assumption, \(G:K\to H^{s'}\) is continuous. Therefore, by
Theorem~\ref{thm:universalapproximation}, for every \(\delta>0\), there exists
an FNO \(\mathcal G_\theta\) such that
\[
    \sup_{\mathbf U^n\in K}
    \|G(\mathbf U^n)-\mathcal G_\theta(\mathbf U^n)\|_{H^{s'}}
    \le \delta .
\]
We obtain
\[
\begin{aligned}
    \|R_\theta(\mathbf U^n)-\Phi_h(\mathbf U^n)\|_{H^{s'}}
    &=
    \|U^n+\mathcal G_\theta(\mathbf U^n)
    -[U^n+G(\mathbf U^n)]\|_{H^{s'}}  =
    \|\mathcal G_\theta(\mathbf U^n)-G(\mathbf U^n)\|_{H^{s'}} .
\end{aligned}
\]
Therefore,
\[
\begin{aligned}
    \sup_{\mathbf U^n\in K}
    \|R_\theta(\mathbf U^n)-\Phi_h(\mathbf U^n)\|_{H^{s'}}
    &=
    \sup_{\mathbf U^n\in K}
    \|\mathcal G_\theta(\mathbf U^n)-G(\mathbf U^n)\|_{H^{s'}} \le \delta .
\end{aligned}
\]
This proves the result.
\end{proof}

\begin{lemma}[\textbf{Lipschitz continuity of the damped projection}]
\label{lem:projection_lipschitz}
Let \(\mathcal K\subset H^{s'}(\mathbb T^d;\mathbb R^{d_u})\) be compact, and define the damped projection
\[
    \Pi_{\mathbf c^\star}^{\eta}(V)
    =
    V+\eta\sum_{j=1}^{q}\lambda_j(V)D_j(V),
    \qquad 0<\eta\le 1 .
\]
Assume that, for each \(j=1,\ldots,q\), the coefficient \(\lambda_j:\mathcal K\to\mathbb R\) is well-defined, uniformly bounded, and Lipschitz continuous on \(\mathcal K\), and that \(D_j:\mathcal K\to H^{s'}(\mathbb T^d;\mathbb R^{d_u})\) is uniformly bounded and Lipschitz continuous on \(\mathcal K\), with lipschitz constants \(M_\lambda,L_\lambda,M_D,L_D>0\) such that, for all \(V,W\in\mathcal K\) respectively.
Then \(\Pi_{\mathbf c^\star}^{\eta}\) is Lipschitz continuous on \(\mathcal K\), with lipschitz constants $ L_\Pi^\eta$
where
$
    L_\Pi^\eta
    =
    1+\eta q(M_\lambda L_D+L_\lambda M_D).
$
\end{lemma}

\begin{proof}
For \(V,W\in\mathcal K\), we have
\[
\begin{aligned}
    \Pi_{\mathbf c^\star}^{\eta}(V)-\Pi_{\mathbf c^\star}^{\eta}(W)
    &=
    V-W
    +
    \eta\sum_{j=1}^{q}
    \left[
    \lambda_j(V)D_j(V)-\lambda_j(W)D_j(W)
    \right].
\end{aligned}
\]
For each \(j\), we observe that 
\[
    \lambda_j(V)D_j(V)-\lambda_j(W)D_j(W)
    =
    \lambda_j(V)\bigl(D_j(V)-D_j(W)\bigr)
    +
    \bigl(\lambda_j(V)-\lambda_j(W)\bigr)D_j(W).
\]
Therefore,
\[
\begin{aligned}
    \|\Pi_{\mathbf c^\star}^{\eta}(V)-\Pi_{\mathbf c^\star}^{\eta}(W)\|_{H^{s'}}
    &\le
    \|V-W\|_{H^{s'}}
    +
    \eta\sum_{j=1}^{q}
    \Bigl(
    |\lambda_j(V)|\|D_j(V)-D_j(W)\|_{H^{s'}} \\
    &\qquad\qquad\qquad
    +
    |\lambda_j(V)-\lambda_j(W)|\|D_j(W)\|_{H^{s'}}
    \Bigr) \\
    &\le
    \left[
    1+\eta q(M_\lambda L_D+L_\lambda M_D)
    \right]
    \|V-W\|_{H^{s'}} .
\end{aligned}
\]
Thus \(\Pi_{\mathbf c^\star}^{\eta}\) is Lipschitz continuous on \(\mathcal K\).
\end{proof}
\begin{lemma}[\textbf{Lipschitz continuity of the EP-FNO one-step map}]
\label{lem:epfno_map_lipschitz}
Let \(K\subset H^s(\mathbb T^d;\mathbb R^{d_a})\) be compact. Assume that \(\mathcal G_\theta\) is Lipschitz.
Then the residual update \(R_\theta\) is Lipschitz on \(K\), with lipschitz constants  $1+L_\mathcal G$.
If, in addition, the projection \(\Pi_{\mathbf c^\star}^{\eta}\) satisfies
Lemma~\ref{lem:projection_lipschitz}, then the EP-FNO one-step map
$
    \Psi_h^\eta(\mathbf U)
    =
    \Pi_{\mathbf c^\star}^{\eta}
    \left(
    R_\theta(\mathbf U)
    \right)
$
is Lipschitz on \(K\) with lipschitz constant  $L_\Psi^\eta$ 
where
$
    L_\Psi^\eta
    =
    L_\Pi^\eta(1+L_{\mathcal G}).
$
\end{lemma}

\begin{proof}
For \(\mathbf U,\mathbf W\in K\), we have

First, define the history-window norm by
\[
\begin{aligned}
    \|\mathbf U-\mathbf W\|_{H^s}
    &=
    \left(
    \sum_{\ell=0}^{T_{\rm in}-1}
    \|U^{n-\ell}-W^{n-\ell}\|_{H^s}^2
    \right)^{1/2}.
\end{aligned}
\]
Since the most recent state corresponds to one term in the sum, we have
\[
\begin{aligned}
    \|U^n-W^n\|_{H^s}^2
    &\le
    \sum_{\ell=0}^{T_{\rm in}-1}
    \|U^{n-\ell}-W^{n-\ell}\|_{H^s}^2, \implies \|U^n-W^n\|_{H^s}
    &\le
    \|\mathbf U-\mathbf W\|_{H^s}.
\end{aligned}
\]
Therefore, using Lemma~\ref{lem:projection_lipschitz}, we obtain
\[
\begin{aligned}
    \|\Psi_h^\eta(\mathbf U)-\Psi_h^\eta(\mathbf W)\|_{H^{s'}}
    &=
    \|\Pi_{\mathbf c^\star}^{\eta}(R_\theta(\mathbf U))
    -
    \Pi_{\mathbf c^\star}^{\eta}(R_\theta(\mathbf W))\|_{H^{s'}} \\
    &\le
    L_\Pi^\eta
    \|R_\theta(\mathbf U)-R_\theta(\mathbf W)\|_{H^{s'}} \\
    &=
    L_\Pi^\eta
    \|U+\mathcal G_\theta(\mathbf U)
    -
    [W+\mathcal G_\theta(\mathbf W)]\|_{H^{s'}} \\
    &\le
    L_\Pi^\eta
    \left(
    \|U-W\|_{H^{s'}}
    +
    \|\mathcal G_\theta(\mathbf U)-\mathcal G_\theta(\mathbf W)\|_{H^{s'}}
    \right) \\
    &\le
    L_\Pi^\eta
    \left(
    \|\mathbf U-\mathbf W\|_{H^s}
    +
    L_{\mathcal G}\|\mathbf U-\mathbf W\|_{H^s}
    \right) \\
    &=
    L_\Pi^\eta(1+L_{\mathcal G})
    \|\mathbf U-\mathbf W\|_{H^s}.
\end{aligned}
\]
as desired. 
\end{proof}

\begin{theorem}[\textbf{Universal approximation by EP-FNO}]
\label{thm:epfno_universal_approximation}
Let \(s,s'\ge 0\), and let \(\mathbf U^n\) denote the history window whose
most recent state is \(U^n\). Let
$
\Phi_h:H^s(\mathbb T^d;\mathbb R^{d_a})
\to H^{s'}(\mathbb T^d;\mathbb R^{d_u})
$
be the reference one-step flow map on a compact set
\(K\subset H^s(\mathbb T^d;\mathbb R^{d_a})\), so that
$
    \Phi_h(\mathbf U^n)=U^{n+1}.
$
Assume that \(\Pi_{\mathbf c^\star}^{\eta}\) satisfies the hypotheses of
Lemma~\ref{lem:projection_lipschitz} on a compact set
\(\mathcal K\subset H^{s'}(\mathbb T^d;\mathbb R^{d_u})\) containing the
relevant provisional states \(R_\theta(\mathbf U^n)\) and reference states
\(\Phi_h(\mathbf U^n)\). Assume also that the projection fixes the reference flow,
$
    \Pi_{\mathbf c^\star}^{\eta}(\Phi_h(\mathbf U^n))
    =
    \Phi_h(\mathbf U^n), \text{for } \mathbf U^n\in K.
$
Then, for every \(\varepsilon>0\), there exists an FNO
$
\mathcal G_\theta:H^s(\mathbb T^d;\mathbb R^{d_a})
\to H^{s'}(\mathbb T^d;\mathbb R^{d_u})
$
such that the EP-FNO map $
    \Psi_h^\eta(\mathbf U^n)
    =
    \Pi_{\mathbf c^\star}^{\eta}
    \left(
    U^n+\mathcal G_\theta(\mathbf U^n)
    \right)
$
satisfies
\[
    \sup_{\mathbf U^n\in K}
    \|\Psi_h^\eta(\mathbf U^n)-\Phi_h(\mathbf U^n)\|_{H^{s'}}
    \le \varepsilon .
\]
\end{theorem}

\begin{proof}
By Lemma~\ref{lem:residual_fno_approximation}, for every \(\delta>0\), there
exists an FNO \(\mathcal G_\theta\) such that the residual update map
satisfies
\[
    \sup_{\mathbf U^n\in K}
    \|R_\theta(\mathbf U^n)-\Phi_h(\mathbf U^n)\|_{H^{s'}}
    \le \delta .
\]
Therefore, using Lemma~\ref{lem:projection_lipschitz},
\[
\begin{aligned}
    \|\Psi_h^\eta(\mathbf U^n)-\Phi_h(\mathbf U^n)\|_{H^{s'}}
    &=
    \|\Pi_{\mathbf c^\star}^{\eta}(R_\theta(\mathbf U^n))
    -
    \Pi_{\mathbf c^\star}^{\eta}(\Phi_h(\mathbf U^n))\|_{H^{s'}} \le
    L_\Pi^\eta
    \|R_\theta(\mathbf U^n)-\Phi_h(\mathbf U^n)\|_{H^{s'}} .
\end{aligned}
\]
Taking the supremum over \(\mathbf U^n\in K\), we obtain
\[
\begin{aligned}
    \sup_{\mathbf U^n\in K}
    \|\Psi_h^\eta(\mathbf U^n)-\Phi_h(\mathbf U^n)\|_{H^{s'}}
    &\le
    L_\Pi^\eta
    \sup_{\mathbf U^n\in K}
    \|R_\theta(\mathbf U^n)-\Phi_h(\mathbf U^n)\|_{H^{s'}} \le
    L_\Pi^\eta \delta .
\end{aligned}
\]
Choose $\delta=\dfrac{\varepsilon}{L_\Pi^\eta}.$ Then
\[
    \sup_{\mathbf U^n\in K}
    \|\Psi_h^\eta(\mathbf U^n)-\Phi_h(\mathbf U^n)\|_{H^{s'}}
    \le \varepsilon .
\]
Hence, the EP-FNO map universally approximates the reference one-step flow map
on \(K\).
\end{proof}

\begin{lemma}[\textbf{Taylor-consistent residual approximation}]
\label{lem:taylor_consistent_residual}
Let \(K\subset H^s(\mathbb T^d;\mathbb R^{d_a})\) be compact. Assume that the
reference one-step flow \(\Phi_h\) admits the Taylor expansion
\[
    \Phi_h(\mathbf U)
    =
    U+\sum_{r=1}^{p} h^r A_r(\mathbf U)
    +
    O(h^{p+1})
\]
uniformly for \(\mathbf U\in K\), where \(U\) is the most recent state in the
history window and each coefficient map
$ A_r:K\to H^{s'}(\mathbb T^d;\mathbb R^{d_u})$
is continuous. Define
$ G_{p,h}(\mathbf U)
    =
    \sum_{r=1}^{p} h^r A_r(\mathbf U).$
Then, for every \(h>0\) sufficiently small, there exists an FNO
\(\mathcal G_{\theta(h)}\) ~\ref{thm:universalapproximation} such that
\[
    \sup_{\mathbf U\in K}
    \|\mathcal G_{\theta(h)}(\mathbf U)-G_{p,h}(\mathbf U)\|_{H^{s'}}
    \le
    C h^{p+1}.
\]
Consequently, the residual update
satisfies
\[
    \sup_{\mathbf U\in K}
    \|R_{\theta(h)}(\mathbf U)-\Phi_h(\mathbf U)\|_{H^{s'}}
    =
    O(h^{p+1}).
\]
\end{lemma}

\begin{proof}
By the Taylor expansion assumption, there exists a constant \(C_1>0\) such
that
\[
    \sup_{\mathbf U\in K}
    \left\|
    \Phi_h(\mathbf U)
    -
    U
    -
    G_{p,h}(\mathbf U)
    \right\|_{H^{s'}}
    \le
    C_1h^{p+1}.
\]
Since each \(A_r\) is continuous on \(K\), the map \(G_{p,h}\) is continuous on
\(K\). Since \(G_{p,h}\) is continuous on the compact set \(K\), the universal
approximation theorem for FNOs implies that, for every tolerance
\(\delta>0\), there exists an FNO \(\mathcal G_\theta\) such that
\[
    \sup_{\mathbf U\in K}
    \|\mathcal G_\theta(\mathbf U)-G_{p,h}(\mathbf U)\|_{H^{s'}}
    \le
    \delta .
\]
For each fixed \(h>0\), choose
\[
    \delta=h^{p+1}.
\]
Then there exists an FNO, denoted by \(\mathcal G_{\theta(h)}\), such that
\[
    \sup_{\mathbf U\in K}
    \|\mathcal G_{\theta(h)}(\mathbf U)-G_{p,h}(\mathbf U)\|_{H^{s'}}
    \le
    h^{p+1}.
\]
Using the triangle inequality, we obtain
\[
\begin{aligned}
    \|R_{\theta(h)}(\mathbf U)-\Phi_h(\mathbf U)\|_{H^{s'}}
    &=
    \|U+\mathcal G_{\theta(h)}(\mathbf U)-\Phi_h(\mathbf U)\|_{H^{s'}} \\
    &\le
    \|\mathcal G_{\theta(h)}(\mathbf U)-G_{p,h}(\mathbf U)\|_{H^{s'}} \\
    &\quad+
    \|U+G_{p,h}(\mathbf U)-\Phi_h(\mathbf U)\|_{H^{s'}} \\
    &\le
    (C_1+1)h^{p+1}.
\end{aligned}
\]
Taking the supremum over \(\mathbf U\in K\) gives
\[
    \sup_{\mathbf U\in K}
    \|R_{\theta(h)}(\mathbf U)-\Phi_h(\mathbf U)\|_{H^{s'}}
    =
    O(h^{p+1}).
\]
\end{proof}

\begin{corollary}[\textbf{Taylor-consistent EP-FNO one-step error}]
\label{cor:taylor_consistent_epfno}
Assume that the hypotheses of Lemma~\ref{lem:taylor_consistent_residual} hold.
Assume also that the projection \(\Pi_{\mathbf c^\star}^{\eta}\) satisfies
Lemma~\ref{lem:projection_lipschitz} and fixes the reference flow:
$\Pi_{\mathbf c^\star}^{\eta}(\Phi_h(\mathbf U))
    =
    \Phi_h(\mathbf U),
 \mathbf U\in K.$
Then the EP-FNO one-step approximation error satisfies
\[
 \sup_{\mathbf U\in K}
    \|\Psi_h^\eta(\mathbf U)-\Phi_h(\mathbf U)\|_{H^{s'}}
    =
    O(h^{p+1}).
\]
\end{corollary}

\begin{proof}
By Lemma~\ref{lem:projection_lipschitz},
\[
\begin{aligned}
    \|\Psi_h^\eta(\mathbf U)-\Phi_h(\mathbf U)\|_{H^{s'}}
    &=
    \|\Pi_{\mathbf c^\star}^{\eta}(R_{\theta(h)}(\mathbf U))
    -
    \Pi_{\mathbf c^\star}^{\eta}(\Phi_h(\mathbf U))\|_{H^{s'}} \\
    &\le
    L_\Pi^\eta
    \|R_{\theta(h)}(\mathbf U)-\Phi_h(\mathbf U)\|_{H^{s'}}.
\end{aligned}
\]
By Lemma~\ref{lem:taylor_consistent_residual},
\[
    \sup_{\mathbf U\in K}
    \|R_{\theta(h)}(\mathbf U)-\Phi_h(\mathbf U)\|_{H^{s'}}
    =
    O(h^{p+1}).
\]
Hence,
\[
\sup_{\mathbf U\in K}\|\Psi_h^\eta(\mathbf U)-\Phi_h(\mathbf U)\|_{H^{s'}}
    =
    O(h^{p+1}).
\]
\end{proof}

\begin{theorem}[\textbf{stability of EP-FNO}]
\label{thm:epfno_finite_time_rollout_stability}
Assume that the hypotheses of
Corollary~\ref{cor:taylor_consistent_epfno} hold, so that the EP-FNO
one-step map satisfies
$
    \sup_{\mathbf U\in K}
    \|\Psi_h^\eta(\mathbf U)-\Phi_h(\mathbf U)\|_{H^{s'}}
    =
    O(h^{p+1}) .
$
Assume also that the hypotheses of
Lemma~\ref{lem:projection_lipschitz} and
Lemma~\ref{lem:epfno_map_lipschitz} hold, so that the full EP-FNO map is
Lipschitz on \(K\) with constant
$L_\Psi^\eta=L_\Pi^\eta(1+L_{\mathcal G})$.
Suppose, in addition, that there exists \(C>0\) such that
$L_\Psi^\eta\le 1+Ch$.
Let
$
    U^{n+1}=\Phi_h(\mathbf U^n)
$
denote the reference rollout, and let
$
    \widehat U^{n+1}
    =
    \Psi_h^\eta(\widehat{\mathbf U}^n)
$
denote the EP-FNO autoregressive rollout. Suppose that both rollouts remain in
the compact rollout set \(K\) for all \(0\le t_n\le T\), and that they start
from the same initial history,
$
    \widehat{\mathbf U}^0=\mathbf U^0 .
$
Then, for every \(t_n=nh\le T\),
\[
    \|\widehat U^n-U^n\|_{H^{s'}}
    =
    O(h^p).
\]
Consequently, the EP-FNO rollout converges to the reference rollout uniformly
on every finite time interval \(0\le t_n\le T\) as \(h\to 0\).
\end{theorem}

\begin{proof}
Let
\[
    e_n=\|\widehat U^n-U^n\|_{H^{s'}} .
\]
Using the definitions of the reference and EP-FNO updates, we obtain
\[
\begin{aligned}
    e_{n+1}
    &=
    \|\widehat U^{n+1}-U^{n+1}\|_{H^{s'}} \\
    &=
    \|\Psi_h^\eta(\widehat{\mathbf U}^n)-\Phi_h(\mathbf U^n)\|_{H^{s'}} \\
    &\le
    \|\Psi_h^\eta(\widehat{\mathbf U}^n)
    -
    \Psi_h^\eta(\mathbf U^n)\|_{H^{s'}}
    +
    \|\Psi_h^\eta(\mathbf U^n)-\Phi_h(\mathbf U^n)\|_{H^{s'}} .
\end{aligned}
\]
By Lemma~\ref{lem:epfno_map_lipschitz}, the EP-FNO map has Lipschitz constant
$L_\Psi^\eta=L_\Pi^\eta(1+L_{\mathcal G})$. Since
$L_\Psi^\eta\le 1+Ch$, and since the one-step approximation error is
\(O(h^{p+1})\), we have
\[
    e_{n+1}
    \le
    (1+Ch)e_n+O(h^{p+1}) .
\]
Since the initial histories agree, \(e_0=0\). Iterating gives
\[
    e_n
    \le
    O(h^{p+1})\sum_{k=0}^{n-1}(1+Ch)^k
    =
    O(h^{p+1})
    \frac{(1+Ch)^n-1}{Ch}.
\]
For \(t_n=nh\le T\), we have
\[
    (1+Ch)^n\le e^{Cnh}\le e^{CT}.
\]
Therefore,
\[
    e_n
    \le
    O(h^{p+1})
    \frac{e^{CT}-1}{Ch}
    =
    O(h^p).
\]
This proves the finite-time rollout stability estimate and the uniform
finite-time convergence as \(h\to 0\).
\end{proof}

\section{Numerical Experiments}
\label{sec:numerical_experiments}

\subsection{Experimental Setup}
\label{subsec:experimental_setup}

In this section, we compare the proposed EP-FNO framework with the standard FNO model on the
Zakharov--Kuznetsov, Kadomtsev--Petviashvili, and Sine--Gordon equations.

The training data are generated from parameterized families of exact solutions. For ZK and KP, we use line-soliton and cylindrical pulse profiles with varying wave parameters, whereas for Sine--Gordon we use kink--antikink and traveling domain-wall solutions.

\subsubsection{Training configuration}\label{trainconfig}

The baseline model is a two-dimensional FNO with four spectral layers, where each layer combines spectral convolution, pointwise linear transformation, and GELU activation. EP-FNO adapts this model by retaining the residual update and applying an invariant projection after each prediction made by the FNO framework.

For first-order-in-time equations such as ZK and KP, the
network predicts the scalar field $u$, while for the sine--Gordon equation we write the system in first-order form and predict the coupled state $(u,v)$,
where $v=u_t$. We intend to learn the one-step solution operator 
\begin{align}
    \mathcal G^{\dagger}:
\bigl(u^{n-T_{\rm in}+1},\ldots,u^n\bigr)
\mapsto u^{n+1}.
\end{align}
Long-time predictions are then obtained by autoregressive rollout. In all our experiments, the models use $T_{\rm in}=10$ previous snapshots as input, $12$ retained Fourier modes,
channel width $24$, batch size $1$, and Adam optimization with weight decay.
A step learning-rate scheduler reduces the learning rate by a factor of $0.5$
every $8$ epochs.


For the parametric experiments, each realization is defined by one sampled parameter tuple. From each realization, we extract \(15\) temporal windows to
form autoregressive input--output pairs. Unless otherwise stated, the dataset contains \(40\) parameter realizations, giving \(600\) samples in total. We split the data by realization so that all windows associated with a given parameter tuple remain in the same split. Specifically, \(32\) realizations are used for training, \(4\) for validation, and \(4\) for testing, corresponding to \(480\), \(60\), and \(60\) temporal windows, respectively. Both models are trained for \(35\) epochs, and the checkpoint with the lowest validation loss is used for testing. All experiments are run on a single NVIDIA A40 GPU with $8$ CPU cores and $96$ GB RAM.

\subsection{Zakharov--Kuznetsov Equation}
\label{app:zk_line_soliton}
Let $T >0$ and $\Omega\subset\mathbb{R}^2$. Consider the two-dimensional Zakharov--Kuznetsov (ZK) equation, 
given by
\begin{equation}
\begin{aligned}
u_t + u u_x + \varepsilon\left(u_{xxx}+u_{xyy}\right) &= 0,
&& (x,y)\in\Omega,\quad t\in(0,T],\\
u(x,y,0) &= u_0(x,y),
&& (x,y)\in\Omega,\\
u(x+L_x,y,t) &= u(x,y,t),
&& y\in[0,L_y],\quad t\in[0,T],\\
u(x,y+L_y,t) &= u(x,y,t),
&& x\in[0,L_x],\quad t\in[0,T],
\end{aligned}
\label{eq:zk_equation_fno_style}
\end{equation}
where \(\varepsilon>0\) is the dispersion parameter. 
The ZK equation conserves mass and Hamiltonian, which can be written as 
\begin{equation}
    \mathcal M_{\rm ZK}[u]=\int_{\Omega}u(x,y,t)\,dx\,dy,
\end{equation}
and 
\begin{equation}
    \mathcal H_{\rm ZK}[u]
=
\int_{\Omega}
\left[
\frac{\varepsilon}{2}(u_x^2+u_y^2)-\frac16u^3
\right]\,dx\,dy.
\end{equation}
respectively.

\subsubsection*{Line Soliton}
We consider \eqref{eq:zk_equation_fno_style} on the domain \([0,L_x]\times[0,L_y]\) with \(T=3\), \(L_x=L_y=8\), and \(\varepsilon=0.01\).
Each initial condition is drawn from a parameterized family of tilted line-soliton profiles. We write
\(\xi=(c,\theta,x_0,y_0)\), where the parameters are sampled independently as $c\sim \mathcal U(0.75,1.25),\, \theta\sim \mathcal U(-0.08,0.08),\, x_0\sim \mathcal U(1.0,3.0),\, y_0\sim \mathcal U(2.5,5.5).$ The corresponding initial condition is
\begin{equation}
    u_0(x,y;\xi)
    =
    3c\,\sech^2
    \left[
        \frac12
        \sqrt{\frac{c}{\varepsilon}}
        \left(
            (x-x_0)\cos\theta+(y-y_0)\sin\theta
        \right)
    \right],
    \label{eq:zk_line_ic_family}
\end{equation}
and the corresponding exact solution given by 
\begin{equation}
    u(x,y,t;\xi)
    =
    3c\,\sech^2
    \left[
        \frac12
        \sqrt{\frac{c}{\varepsilon}}
        \left(
            (x-ct-x_0)\cos\theta+(y-y_0)\sin\theta
        \right)
    \right].
    \label{eq:zk_line_solution_family}
\end{equation}

\begin{table}[H]
\centering
\caption{Relative $L^2$ error comparison of EP-FNO and FNO solutions against the exact solution.}
\label{tab:epfno_fno_line_soliton_l2_error_comparison}
\renewcommand{\arraystretch}{1.15}
\begin{tabular}{lccccccc}
\toprule[1.5pt]
\textbf{Method}   
& \textbf{$t=0.00$} & \textbf{$t=0.50$} & \textbf{$t=1.00$} & \textbf{$t=1.50$} 
& \textbf{$t=2.00$} & \textbf{$t=2.50$} & \textbf{$t=3.00$} \\
\midrule[1.2pt]
EP-FNO  
& \textbf{0.0000} & \textbf{0.0100} & \textbf{0.0221} 
& \textbf{0.0295} & \textbf{0.0362} & \textbf{0.0375} 
& \textbf{0.0449} \\

FNO  
& \textbf{0.0000} & 0.0446 & 0.0577 & 0.0716 
& 0.0832 & 0.1073 & 0.1545 \\
\bottomrule[1.5pt]
\end{tabular}
\end{table}

Table~\ref{tab:epfno_fno_line_soliton_l2_error_comparison} shows that EP-FNO consistently gives smaller relative $L^2$ errors than the standard FNO for this result. Both models start from the exact initial condition, so the error is zero at $t=0$, but the difference becomes clear as time evolves. At the final time $t=3$, the EP-FNO error is $0.0449$, while the FNO error increases to $0.1545$. Thus, the EP-FNO final-time error is about $3.4$ times smaller than that of the baseline FNO.
Figure~\ref{fig:zk_line} shows that the EP-FNO prediction preserves the main ridge-like
structure of the line soliton over the rollout interval. The wave profile remains sharp
and travels in the correct direction, indicating that the projection step does not
destroy the qualitative soliton dynamics. In addition, Figure~\ref{fig:zklineHandm}
shows that the mass and Hamiltonian remain preserved during the EP-FNO rollout. This
confirms that the projection step helps control invariant drift while maintaining the
correct qualitative propagation of the ZK line soliton.

\subsubsection*{Cylindrical Solitary Pulse}
The second ZK benchmark follows the cylindrically symmetric solitary pulse example used
by Chen et al.~\cite{chen2011multi}. This example is more demanding than the line-soliton
case because the solution is localized in both spatial directions and its evolution
requires the model to preserve the radial pulse structure during propagation.

We consider a solution of \eqref{eq:zk_equation_fno_style} corresponding to cylindrically symmetric solitary pulses on the domain \(L_x=L_y=32\) over the time interval \(T=3\). Let \(\eta=(c,x_0,y_0)\) denote the parameter vector. The initial condition is then written as 

\begin{equation}
u_0(x,y; \xi)
=
\frac{c}{3}
\sum_{n=1}^{10}
a_{2n}
\left[
\cos
\left(
2n\arccot
\left(
\frac{\sqrt{c}}{2}
\sqrt{(x-x_0)^2+(y-y_0)^2}
\right)
\right)
-1
\right].
\label{eq:zk_cylindrical_ic_family}
\end{equation}
The parameters are sampled independently as $c\sim\mathcal U(3.5,4.5),\,x_0\sim\mathcal U(8,12),\, y_0\sim\mathcal U(14,18).$ 

\begin{table}[H]
\centering
\caption{Coefficients \(a_{2n}\) used in the cylindrically symmetric solitary
solution of the ZK equation, following Chen et al.~\cite{chen2011multi}.}
\label{tab:zk_a2n_coefficients}
\renewcommand{\arraystretch}{1.15}
\setlength{\tabcolsep}{4pt}
\resizebox{\linewidth}{!}{%
\begin{tabular}{c|cccccccccc}
\hline
\(n\) & 1 & 2 & 3 & 4 & 5 & 6 & 7 & 8 & 9 & 10 \\
\hline
\(a_{2n}\)
& \(-1.25529873\)
& \(0.21722635\)
& \(0.06452543\)
& \(0.00540862\)
& \(-0.00332515\)
& \(-0.00281281\)
& \(-0.00138352\)
& \(-0.00070289\)
& \(-0.00020451\)
& \(-0.00003053\) \\
\hline
\end{tabular}%
}
\end{table}
\noindent 
The coefficients \(a_{2n}\) appearing in the finite-series representation are
listed in Table~\ref{tab:zk_a2n_coefficients}.

\begin{table}[H]
\centering
\caption{Relative $L^2$ error comparison of EP-FNO and FNO solutions against the exact solution.}
\label{tab:epfno_fno_cylindrical_zk_l2_error_comparison}
\renewcommand{\arraystretch}{1.15}
\begin{tabular}{lccccccc}
\toprule[1.5pt]
\textbf{Method} 
& \textbf{$t=0.00$} & \textbf{$t=0.50$} & \textbf{$t=1.00$} & \textbf{$t=1.50$} 
& \textbf{$t=2.00$} & \textbf{$t=2.50$} & \textbf{$t=3.00$} \\
\midrule[1.2pt]
EP-FNO 
& $\mathbf{0.0000}$ 
& $\mathbf{0.0322}$ 
& $\mathbf{0.0514}$ 
& $\mathbf{0.0755}$ 
& $\mathbf{0.1015}$ 
& $\mathbf{0.1177}$ 
& $\mathbf{0.1175}$ \\
FNO 
& \textbf{$0.0000$} 
& $0.0802$ 
& $0.1485$ 
& $0.2367$ 
& $0.3347$ 
& $0.4101$ 
& $0.4754$ \\
\bottomrule[1.5pt]
\end{tabular}
\end{table}

Table~\ref{tab:epfno_fno_cylindrical_zk_l2_error_comparison} shows that while both models start with comparable short-time errors, the standard FNO error grows steadily to $4.75\times10^{-1}$ at $t=3$, whereas EP-FNO achieves a final error of only $1.18\times10^{-1}$. These results demonstrate that explicit invariant preservation can substantially improve both physical fidelity and long-time prediction accuracy for two-dimensional dispersive wave dynamics.

Figure~\ref{fig:zk_cylindrical_epfno_surface_result} shows that the EP-FNO rollout captures the main localized
pulse and its motion across the domain. The predicted solution captures the same physical behaviour and propagation in \cite{chen2011multi} and does not exhibit severe spreading, which suggests that the
structure-informed projection helps control long-time degradation of the learned
trajectory. In addition, Figure~\ref{fig:zkcyclHandm} shows that the mass and Hamiltonian
are preserved during the rollout, further confirming that the projection step helps
stabilize the learned dynamics while maintaining the qualitative behavior of the
cylindrically symmetric solitary wave.
\begin{figure}[htbp]
    \centering
    \includegraphics[width=0.6\linewidth]{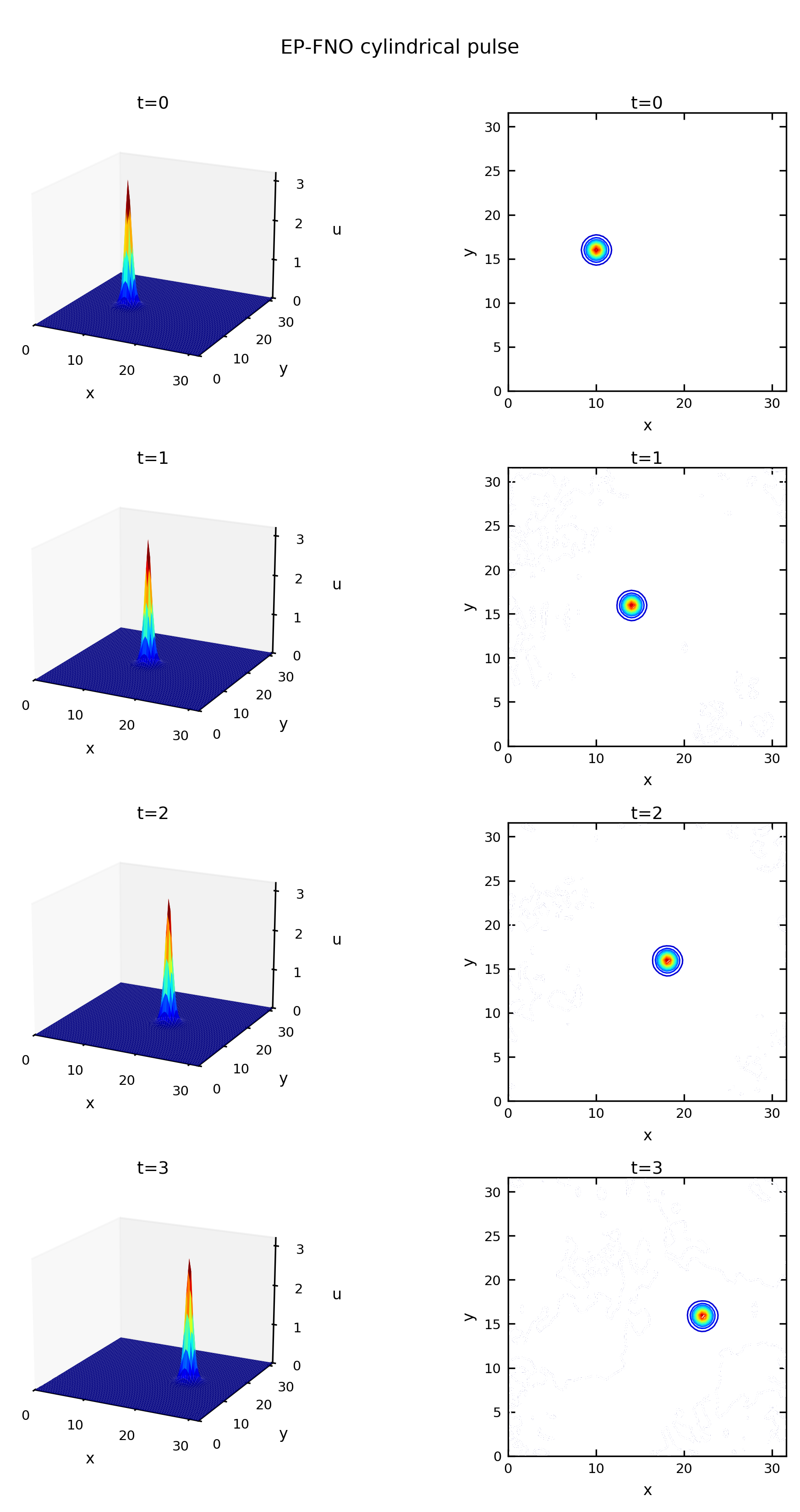}
        \caption{EP-FNO rollout for the Cylindrical ZK pulse benchmark at representative times $t = 0, 1, 2, \text{ and } 3.$ }
        \label{fig:zk_cylindrical_epfno_surface_result}
\end{figure}

\begin{figure}[htbp]
    \centering
    \begin{subfigure}{0.48\textwidth}
        \centering
        \includegraphics[width=\linewidth]{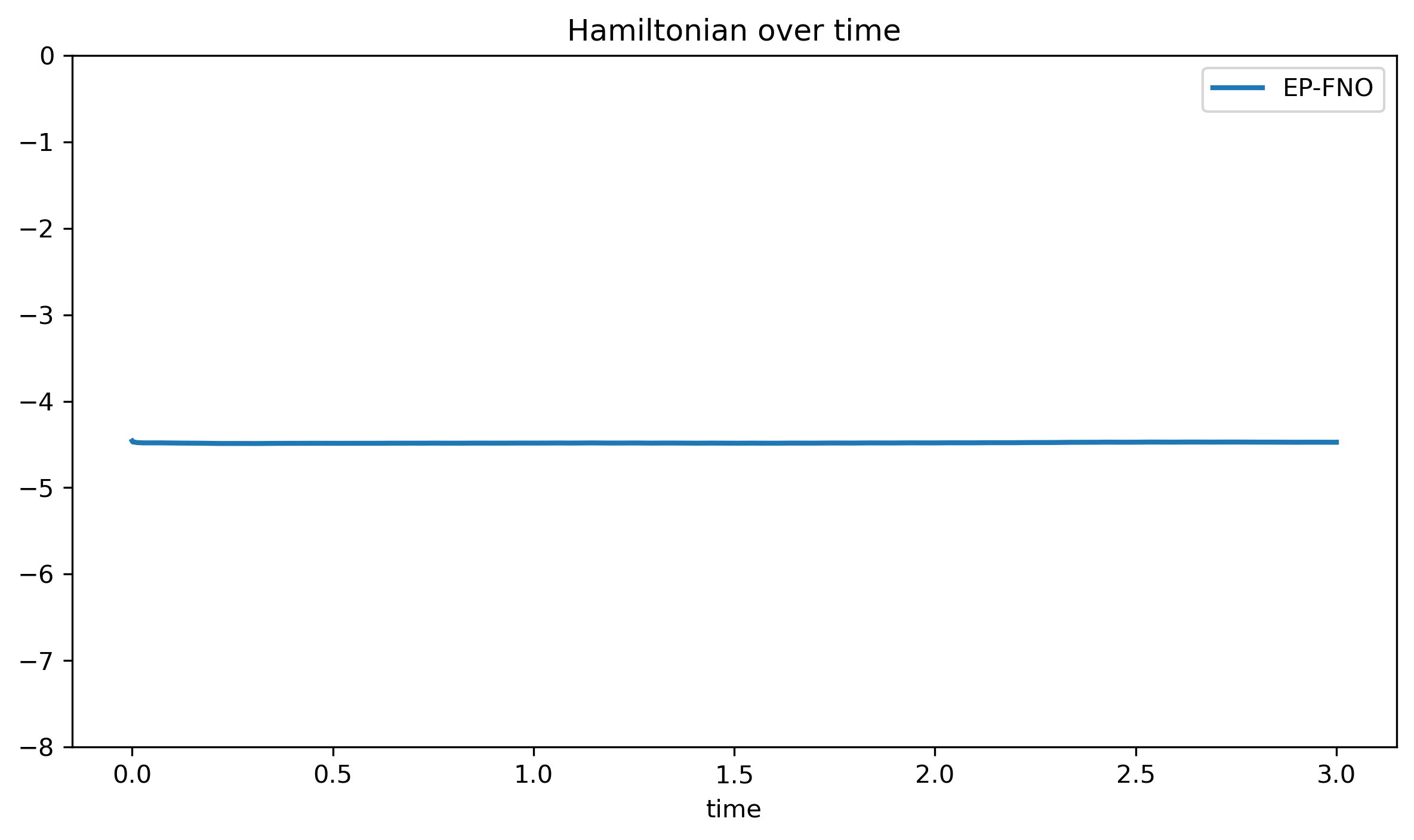}
        \caption{}
        \label{fig:zkcyclh}
    \end{subfigure}
    \hfill
     \begin{subfigure}{0.48\textwidth}
        \centering
        \includegraphics[width=\linewidth]{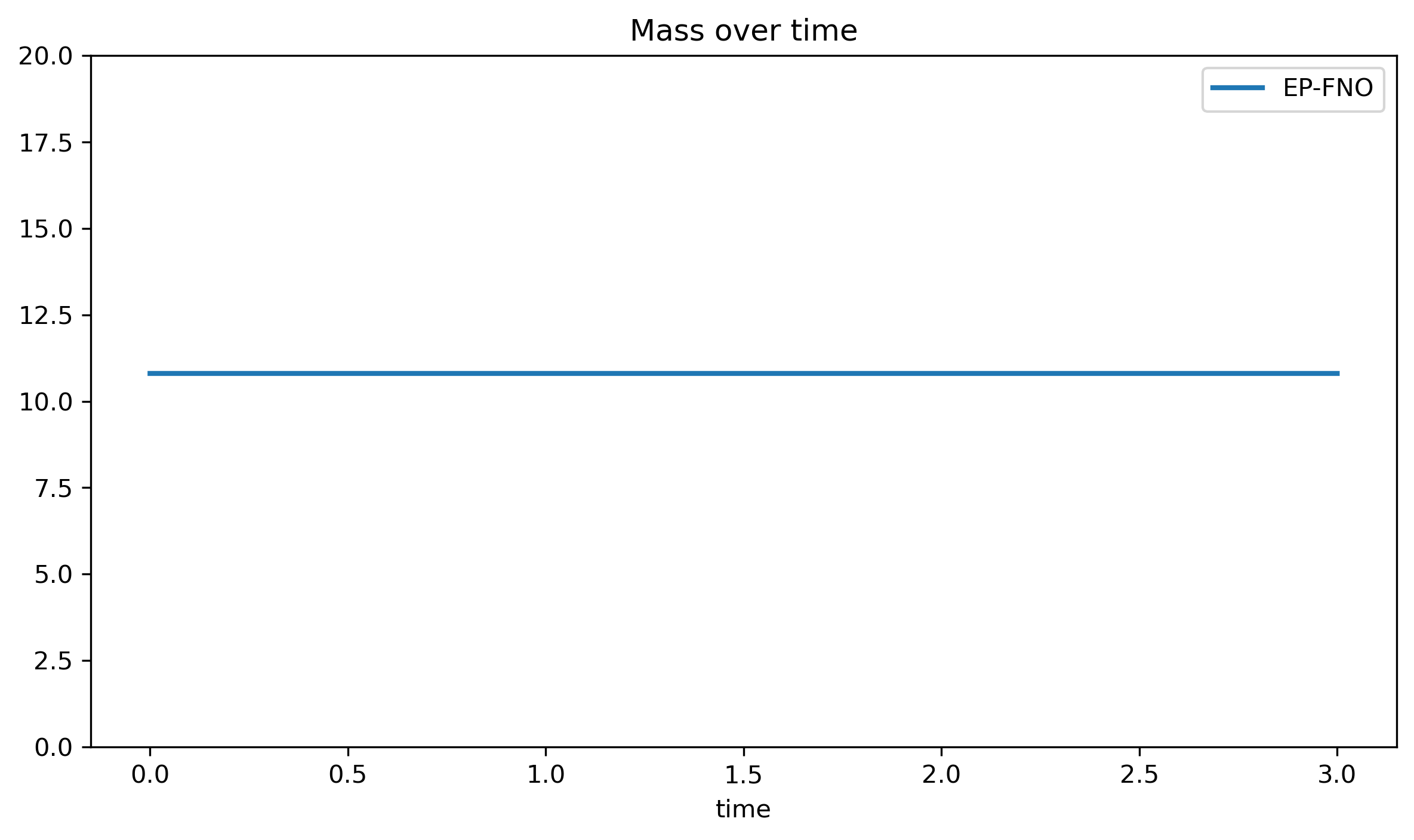}
        \caption{}
        \label{fig:zkcycl}
    \end{subfigure}

    \caption{Evolution of the Hamiltonian and mass during the EP-FNO rollout for the ZK cylindrical pulse benchmark.}
    \label{fig:zkcyclHandm}
\end{figure}

\subsection{Kadomtsev--Petviashvili Equation}
\label{subsubsec:kp_line_results}

Let $T>0$ and $\Omega\subset\mathbb{R}^2$. Consider the two-dimensional Kadomtsev--Petviashvili (KP) equation, given by
\begin{equation}
\begin{aligned}
\left(u_t+(f(u))_x+u_{xxx}\right)_x+\sigma u_{yy} &= 0,
&& (x,y)\in\Omega,\quad t\in(0,T],\\
u(x,y,0) &= u_0(x,y),
&& (x,y)\in\Omega,\\
u(x+L_x,y,t) &= u(x,y,t),
&& y\in[0,L_y],\quad t\in[0,T],\\
u(x,y+L_y,t) &= u(x,y,t),
&& x\in[0,L_x],\quad t\in[0,T].
\end{aligned}
\label{eq:kp_equation_fno_style}
\end{equation}
The KP equation conserves mass and Hamiltonian given by 
\begin{equation}
    \mathcal M_{\rm KP}[u]=\int_{\Omega}u(x,y,t)\,dx\,dy
\end{equation}
and 
\begin{equation}
    \mathcal H_{\rm KP}[u]=\int_{\Omega}
\left[\frac12u_x^2-F(u)-\frac{\sigma}{2}
(\partial_x^{-1}u_y)^2\right]\,dx\,dy
\end{equation}
respectively. In the above, $F'(u)=f(u)$. 
We choose $f(u)=3u^2$, so that $F(u)=u^3$. 

\subsubsection*{Line Soliton}
We consider the KP line-soliton family parameterized by
\(\xi=(k,\lambda,x_0,\sigma)\). The exact solution is given by
\begin{equation}
    u(x,y,t;\xi)
    =
    2k^2
    \sech^2
    \left[
        k
        \left(
            x+\lambda y
            -
            (4k^2+\sigma\lambda^2)t
            -
            x_0
        \right)
    \right].
    \label{eq:kp_line_solution_family}
\end{equation}
The corresponding initial condition is
\begin{equation}
    u_0(x,y;\xi)
    =
    u(x,y,0;\xi)
    =
    2k^2
    \sech^2
    \left[
        k(x+\lambda y-x_0)
    \right].
    \label{eq:kp_line_ic_family}
\end{equation}
For each realization, the parameters are chosen as
\(k\sim \mathcal U(0.90,1.10)\), \(\lambda=0\), \(x_0=10\), and
\(\sigma=-3\).

\begin{table}[H]
\centering
\caption{Relative $L^2$ error comparison of EP-FNO and FNO solutions against the exact solution.}
\label{tab:epfno_fno_kp_line_l2_error_comparison}
\renewcommand{\arraystretch}{1.15}
\begin{tabular}{lccccccc}
\toprule[1.5pt]
\textbf{Method} 
& \textbf{$t=0.00$} & \textbf{$t=0.50$} & \textbf{$t=1.00$} & \textbf{$t=1.50$} 
& \textbf{$t=2.00$} & \textbf{$t=2.50$} & \textbf{$t=3.00$} \\
\midrule[1.2pt]
EP-FNO 
& \textbf{0.0000} & \textbf{0.1622} & \textbf{0.2896} 
& \textbf{0.3764} & \textbf{0.4300} & \textbf{0.5108} 
& \textbf{0.6526} \\
FNO    
& \textbf{0.0000} & 0.3466 & 0.4801 
& 0.7228 & 0.9645 & 0.9724 
& 0.9768 \\
\bottomrule[1.5pt]
\end{tabular}
\end{table}

Table~\ref{tab:epfno_fno_kp_line_l2_error_comparison} shows the relative $L^2$ error between EP-FNO and FNO compared with the exact solution \eqref{eq:kp_line_solution_family}. We can observe that the error of the FNO model grows rapidly as time increases and approaches unity by $t=3$, indicating severe degradation of the predicted solution. In contrast, EP-FNO exhibits a significantly slower growth in error, maintaining substantially better agreement with the exact solution at all reported times. 

\begin{figure}[H]
    \centering
    \includegraphics[width=0.6\linewidth]{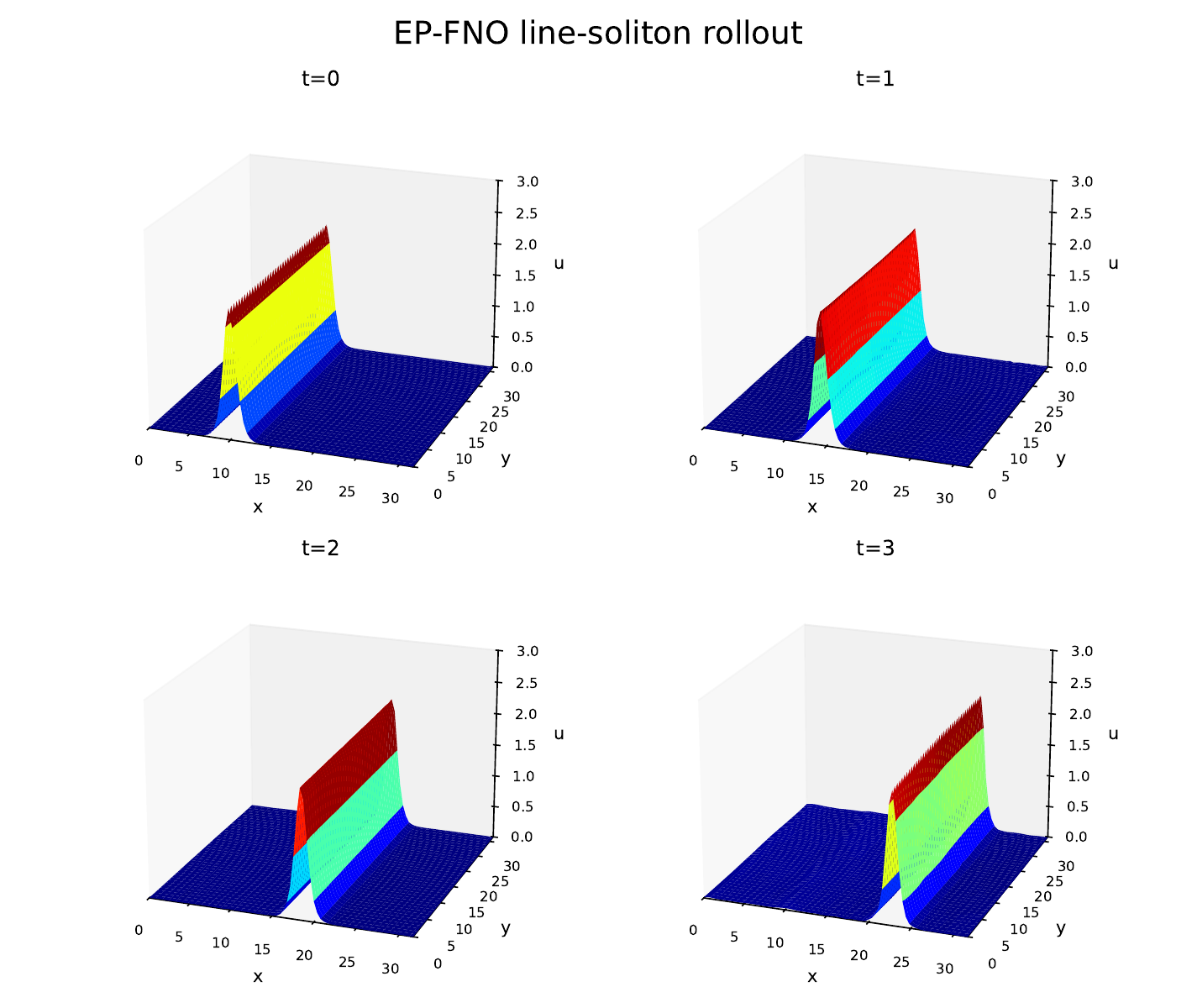}
    \caption{EP-FNO rollout for the KP line-soliton benchmark at representative times \(t=0,1,2,\) and \(3\).}
    \label{fig:kp_epfno_surface_result}
\end{figure}

\begin{figure}[H]
\centering

\begin{subfigure}{0.49\textwidth}
    \centering
    \includegraphics[width=\linewidth]{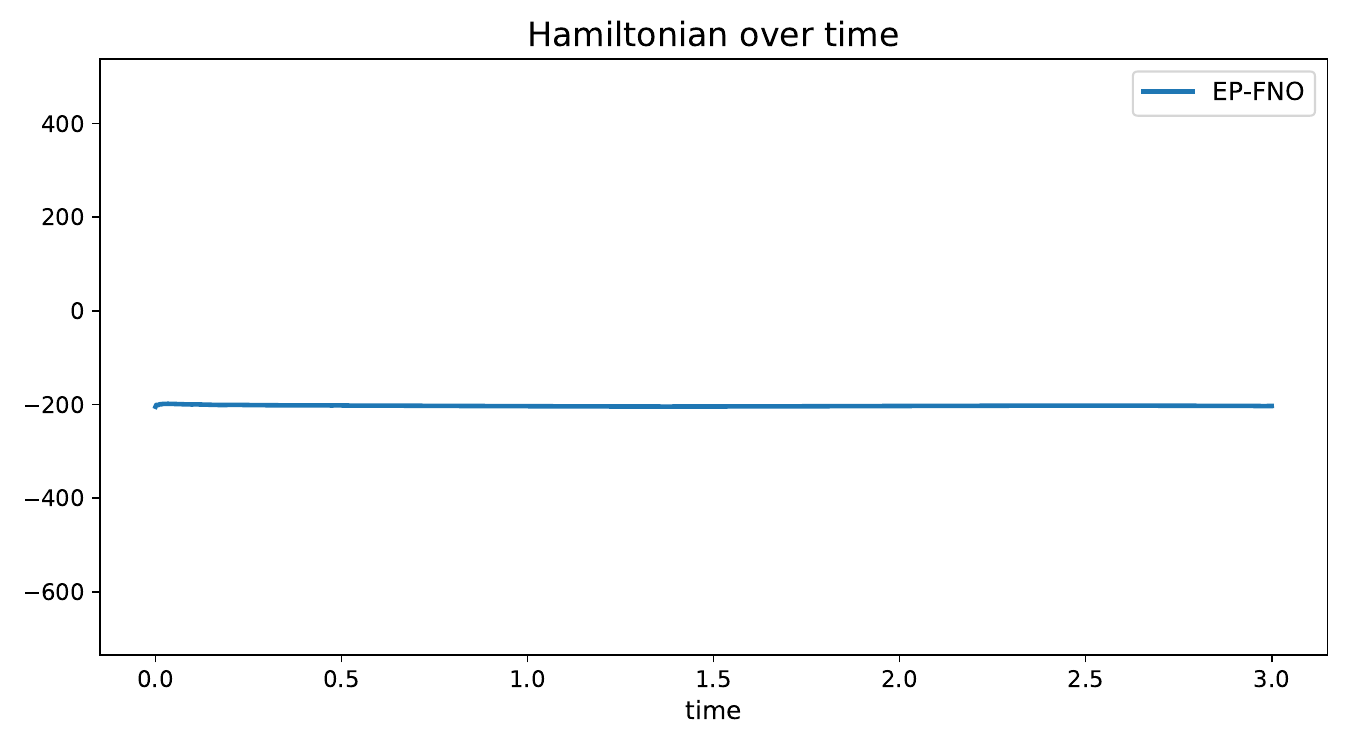}
    \caption{}
\end{subfigure}
\hfill
\begin{subfigure}{0.49\textwidth}
    \centering
    \includegraphics[width=\linewidth]{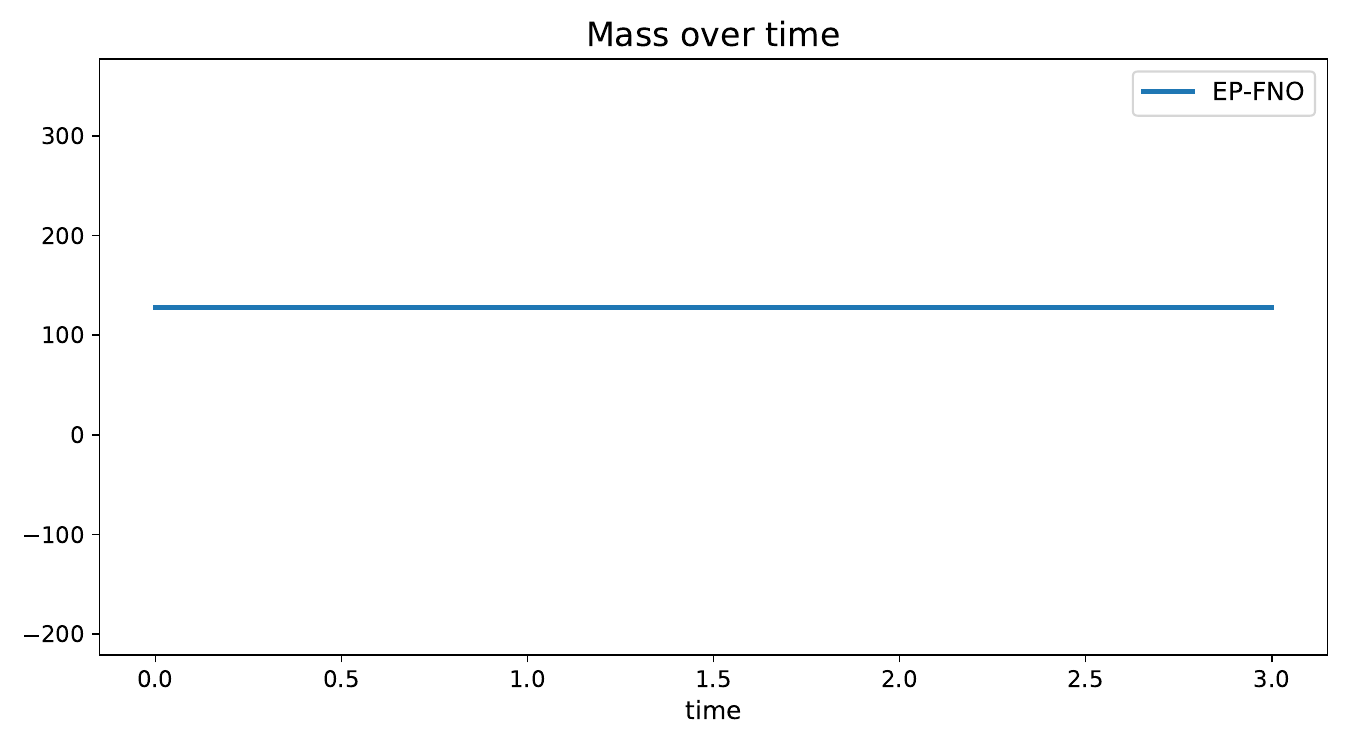}
    \caption{}
\end{subfigure}

\vspace{2mm}

\caption{Evolution of the Hamiltonian and mass during the EP-FNO rollout for the KP line-soliton benchmark.}
\label{fig:kp_lineham_and_mass}
\end{figure}

Figure~\ref{fig:kp_epfno_surface_result} shows the EP-FNO prediction at representative times during the rollout. The line soliton propagates across the computational domain while retaining its coherent structure throughout the prediction horizon. As shown in Figure~\ref{fig:kp_lineham_and_mass}, both the Hamiltonian and the mass remain nearly constant over the entire rollout interval, indicating that the structure-preserving projection effectively controls invariant drift during autoregressive prediction. For comparison, the corresponding FNO rollout and invariant evolution are provided in Appendix~\ref{appenKp}.

\subsection{Sine--Gordon Equation}
\label{subsubsec:sg_results}

Let $T>0$ and $\Omega\subset\mathbb{R}^2$. Consider the two-dimensional sine--Gordon (SG) equation, given by
\begin{equation}
\begin{aligned}
u_{tt}-u_{xx}-u_{yy}+\sin(u) &= 0,
&& (x,y)\in\Omega,\quad t\in(0,T],\\
u(x,y,0) &= u_0(x,y),
&& (x,y)\in\Omega,\\
u_t(x,y,0) &= v_0(x,y),
&& (x,y)\in\Omega,\\
u(x+L_x,y,t) &= u(x,y,t),
&& y\in[0,L_y],\quad t\in[0,T],\\
u(x,y+L_y,t) &= u(x,y,t),
&& x\in[0,L_x],\quad t\in[0,T].
\end{aligned}
\label{eq:sg_2d_equation}
\end{equation}
Setting the velocity variable as $v=u_t$, the SG equation can be written as the first-order system
\begin{equation}
\begin{aligned}
u_t &= v,
&& (x,y)\in\Omega,\quad t\in(0,T],\\
v_t &= u_{xx}+u_{yy}-\sin(u),
&& (x,y)\in\Omega,\quad t\in(0,T],\\
u(x,y,0) &= u_0(x,y),
&& (x,y)\in\Omega,\\
v(x,y,0) &= v_0(x,y),
&& (x,y)\in\Omega,\\
u(x+L_x,y,t) &= u(x,y,t),
&& y\in[0,L_y],\quad t\in[0,T],\\
u(x,y+L_y,t) &= u(x,y,t),
&& x\in[0,L_x],\quad t\in[0,T].
\end{aligned}
\label{eq:sg_2d_first_order_system}
\end{equation}
The sine--Gordon Hamiltonian is
\begin{equation}
    \mathcal H_{\rm SG}[u,v]
    =
    \int_{\Omega}
    \left[
        \frac12 v^2
        +
        \frac12\left(u_x^2+u_y^2\right)
        +
        1-\cos(u)
    \right]\,dx\,dy .
    \label{eq:sg_hamiltonian}
\end{equation}
which is a conserved quantity. 
The Hamiltonian consists of kinetic, gradient, and potential energy contributions,
which we write as
\begin{equation}
    \mathcal H_{\rm SG}[u,v]
    =
    E_{\rm kin}[v]
    +
    E_{\rm grad}[u]
    +
    E_{\rm pot}[u],
    \label{eq:sg_hamiltonian_decomposition}
\end{equation}
where
\begin{equation}
    E_{\rm kin}[v]
    =
    \frac12\int_{\Omega} v^2\,dx\,dy,
    \qquad
    E_{\rm grad}[u]
    =
    \frac12\int_{\Omega}
    \left(u_x^2+u_y^2\right)\,dx\,dy,
    \label{eq:sg_kinetic_gradient_energy}
\end{equation}
and
\begin{equation}
    E_{\rm pot}[u]
    =
    \int_{\Omega}
    \left(1-\cos(u)\right)\,dx\,dy .
    \label{eq:sg_potential_energy}
\end{equation}
In particular, \(E_{\rm grad}[u]\) measures the spatial-gradient energy of the
domain-wall profile and is used to monitor whether the learned rollout preserves
the sharp transition structure of the sine--Gordon solution.

\subsubsection*{Traveling domain-wall}
The traveling domain-wall profile used in this experiment follows the exact two-dimensional sine--Gordon domain-wall solution described by Johnson, Suarez, and Biswas~\cite{johnson2012new}.
We consider a traveling domain-wall family parameterized by
\(\xi=(B_1,B_2,x_0,y_0,t_0)\). Define
\[
    B_3=\sqrt{B_1^2+B_2^2-1},
\]
with \(B_1^2+B_2^2>1\) enforced so that \(B_3\) is real. The exact displacement field is
\begin{equation}
    u(x,y,t;\xi)
    =
    4\arctan\!\left(\exp(\theta(x,y,t;\xi))\right),
    \label{eq:sg_2d_exact_solution}
\end{equation}
where $ \theta(x,y,t;\xi)= B_1(x-x_0)+B_2(y-y_0)-B_3(t-t_0).$ The associated velocity field is given by
\begin{equation}
    v(x,y,t;\xi)
    =
    -2B_3\sech(\theta(x,y,t;\xi)).
    \label{eq:sg_2d_exact_velocity}
\end{equation}
The corresponding initial conditions are
\begin{equation}
    u_0(x,y;\xi)
    =
    4\arctan\!\left(\exp(\theta(x,y,0;\xi))\right),
    \label{eq:sg_2d_initial_displacement}
\end{equation}
and
\begin{equation}
    v_0(x,y;\xi)
    =
    v(x,y,0;\xi)
    =
    -2B_3\sech(\theta(x,y,0;\xi)).
    \label{eq:sg_2d_initial_velocity}
\end{equation}
For each realization, the parameters are sampled independently as $B_1\sim\mathcal U(0.90,1.20), B_2\sim\mathcal U(0.70,1.00),\, x_0\sim\mathcal U(-1.0,1.0),\, y_0\sim\mathcal U(-1.0,1.0),$ and $t_0\sim\mathcal U(-0.5,0.5).$

\begin{table}[H] 
\centering
\caption{Relative \(L^2\) error comparison of EP-FNO and FNO for the SG solution field \(u\).}
\label{tab:epfno_fno_sine_gordon_2d_u_l2_error_comparison}
\renewcommand{\arraystretch}{1.15}
\begin{tabular}{lccccc}
\toprule[1.5pt]
\textbf{Model} 
& \textbf{\(t=0.00\)} & \textbf{\(t=1.00\)} & \textbf{\(t=2.00\)} 
& \textbf{\(t=3.00\)} & \textbf{\(t=4.00\)} \\
\midrule[1.2pt]
EP-FNO 
& \textbf{0.000000} & 0.002221 & \textbf{0.002501} 
& \textbf{0.004175} & \textbf{0.006927} \\
FNO    
& \textbf{0.000000} & \textbf{0.001838} & 0.003447 
& 0.006636 & 0.017220 \\
\bottomrule[1.5pt]
\end{tabular}
\end{table}

Table~\ref{tab:epfno_fno_sine_gordon_2d_u_l2_error_comparison} shows that FNO gives a slightly smaller error at early time \(t=1\), but EP-FNO is more accurate over longer rollouts. In particular, EP-FNO exhibits slower error growth and achieves a substantially smaller error at \(t=4\). This indicates that the Hamiltonian correction improves the long-time stability of the autoregressive predictions.

\begin{figure}[H]
\centering

\begin{subfigure}{0.41\textwidth}
    \centering
    \includegraphics[width=\linewidth]{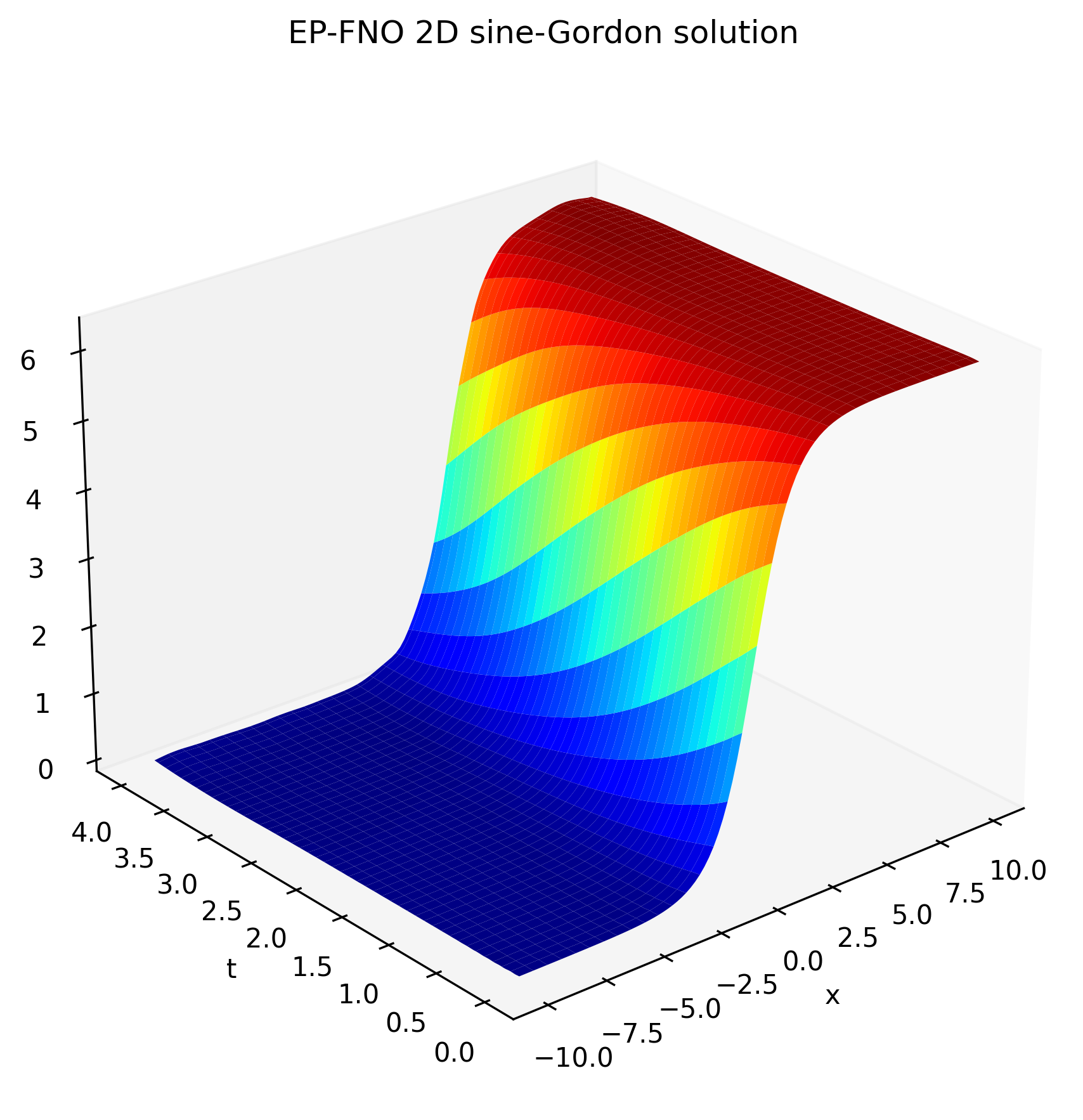}
    \caption{EP-FNO spatiotemporal solution.}
\end{subfigure}
\hfill
\begin{subfigure}{0.53\textwidth}
    \centering
    \includegraphics[width=\linewidth]{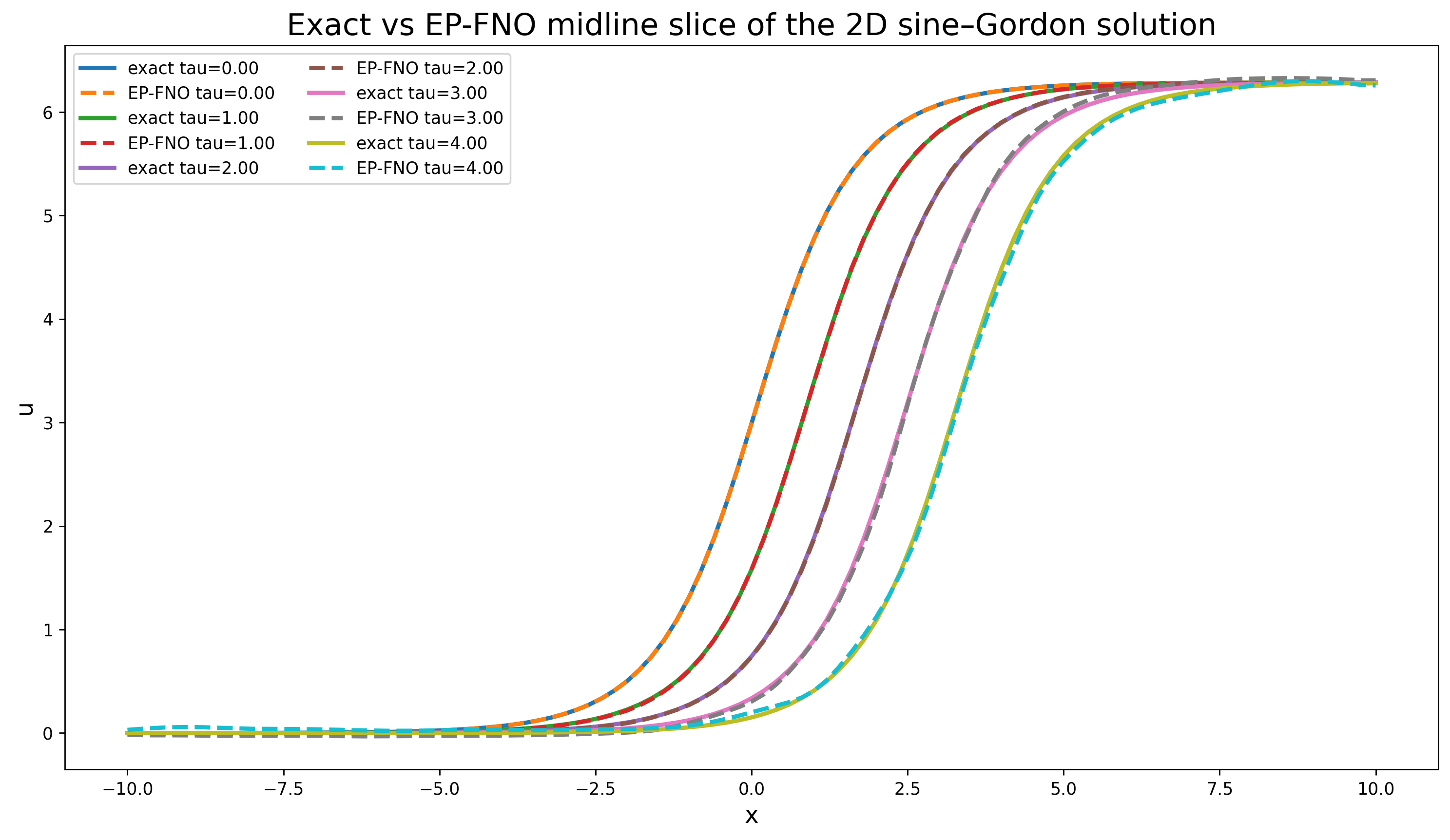}
    \caption{Exact and EP-FNO midline profiles}
\end{subfigure}
\caption{Two-dimensional sine--Gordon experiment. 
    (a) EP-FNO prediction of the spatiotemporal solution profile. 
    (b) Midline comparison between the exact solution and the EP-FNO prediction.}
\label{fig:sgsurfandmid}

\end{figure}

Although the sine--Gordon equation is Hamiltonian at the continuous level, the Hamiltonian computed from our reference data is not exactly conserved \cite{johnson2012new}, so we use the spatial  gradient energy as a diagnostic for long-time stability. 
Figure~\ref{fig:sgsurfandmid} shows that the EP-FNO prediction captures the motion of the domain wall and preserves the main transition structure of the displacement field \(u\). The predicted interface remains well aligned with the reference solution, indicating that the Hamiltonian projection does not distort the qualitative wave dynamics.

\begin{figure}[htbp]

        \centering
        \includegraphics[width=0.5\linewidth]{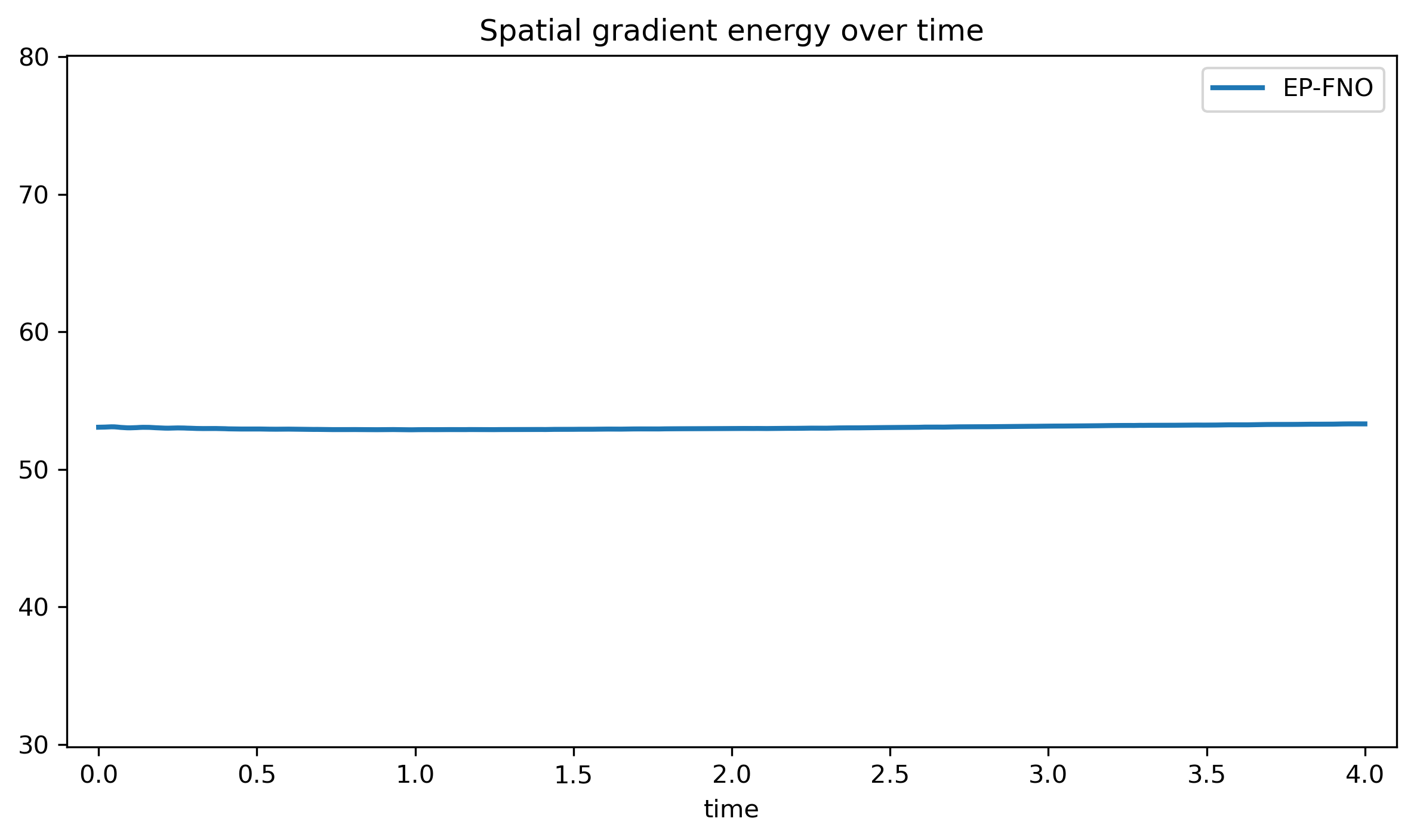}
        \caption{Spatial Gradient Energy plot over time. }
    \label{fig:sggradientenergy}
\end{figure}
Figure~\ref{fig:sggradientenergy} shows the corresponding gradient energy evolution during the rollout, controlling the Hamiltonian drift is important for stable long-time prediction.  The EP-FNO prediction maintains a more
stable gradient-energy profile, indicating that the learned solution preserves
the sharp domain-wall transition without excessive smoothing or artificial
oscillation.

\subsection{Ablation Study}

\subsubsection{Training and Validation Diagnostics}
\label{subsec:ablation_study}

We first examine the contribution of the two main architectural components: the residual time-step update and the invariant projection. We compare three models on the ZK line-soliton benchmark: a standard FNO, the proposed EP-FNO, and an EP-FNO variant in which the residual update is removed while the projection mechanism is retained. 

\begin{figure}[htbp]
\centering
\begin{subfigure}{0.48\textwidth}
    \centering
    \includegraphics[width=\linewidth]{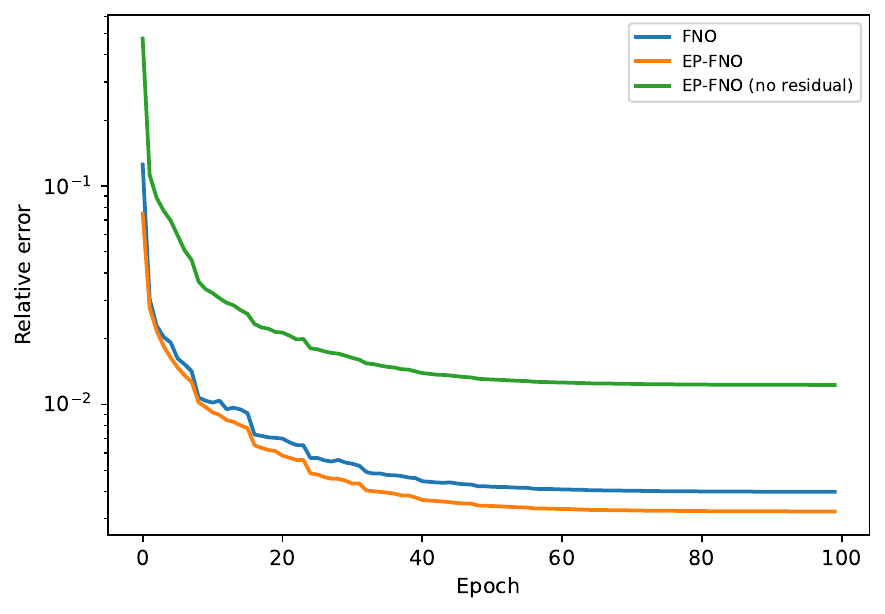}
    \caption{
    }
\end{subfigure}
\hfill
\begin{subfigure}{0.48\textwidth}
    \centering
    \includegraphics[width=\linewidth]{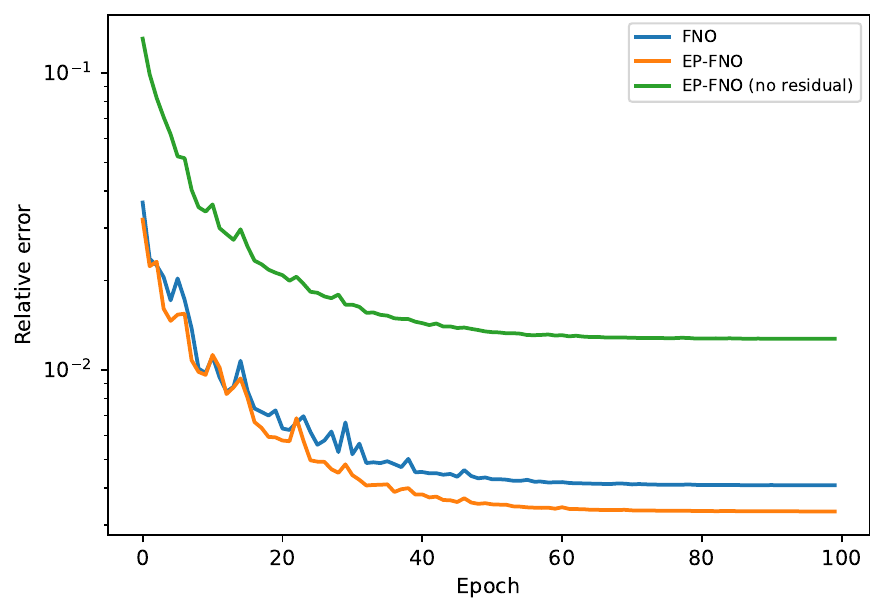}
    \caption{
    }
\end{subfigure}
\vspace{2mm}
\caption{Training (Fig. (a)) and validation (Fig. (b)) relative errors as a function of epochs for FNO, EP-FNO, and EP-FNO without the residual update on the ZK line-soliton problem.}
\label{fig:zk-line-train-eval}
\end{figure}

For the ZK line-soliton problem, EP-FNO attains a best validation error of $3.33\times10^{-3}$, improving upon the baseline FNO error of $4.08\times10^{-3}$. The non-residual EP-FNO variant reaches a substantially higher error of $1.27\times10^{-2}$. 
The non-residual EP-FNO—projection without the residual formulation—consistently underperforms both EP-FNO and the baseline FNO. This indicates that the projection step alone cannot compensate for the lack of a suitable time-stepping inductive bias. Conversely, the residual update provides a natural way to learn small increments in wave propagation, but without projection the model still suffers from invariant drift over long rollouts. Only the combination of the two yields the best accuracy and stability.

The training and validation curves in Figure~\ref{fig:zk-line-train-eval} also show consistent convergence behavior. The residual EP-FNO exhibits stable error reduction, while the non-residual variant plateaus at a higher error level.

In summary, the empirical evidence on the ZK line-soliton benchmark demonstrates that both components of the proposed method contribute to performance. The residual neural-operator update supplies a powerful inductive bias for learning wave dynamics, and the energy projection complements it by enforcing physical invariants. Their combination yields the most accurate and stable model among those tested.

\subsubsection{Cost--Accuracy Tradeoff}
\label{subsec:cost_accuracy_tradeoff}

To further highlight the efficiency of our method, we compare the cost--accuracy behavior of FNO and EP-FNO for the Zakharov--Kuznetsov experiment.
Since the EP-FNO uses the FNO framework, the additional cost of EP-FNO comes from the projection
step. For this experiment, we normalize the FNO cost to $1.00$ and compute the
EP-FNO cost from the measured runtime ratio at the common training horizon.
For ZK, the corresponding runtimes are $704.27$ s for FNO and $875.07$ s for
EP-FNO, giving a normalized EP-FNO cost of $875.07/704.27\approx 1.243$.

\begin{figure}[htbp]
\centering
\begin{subfigure}{0.50\textwidth}
\centering
\includegraphics[width=\linewidth]{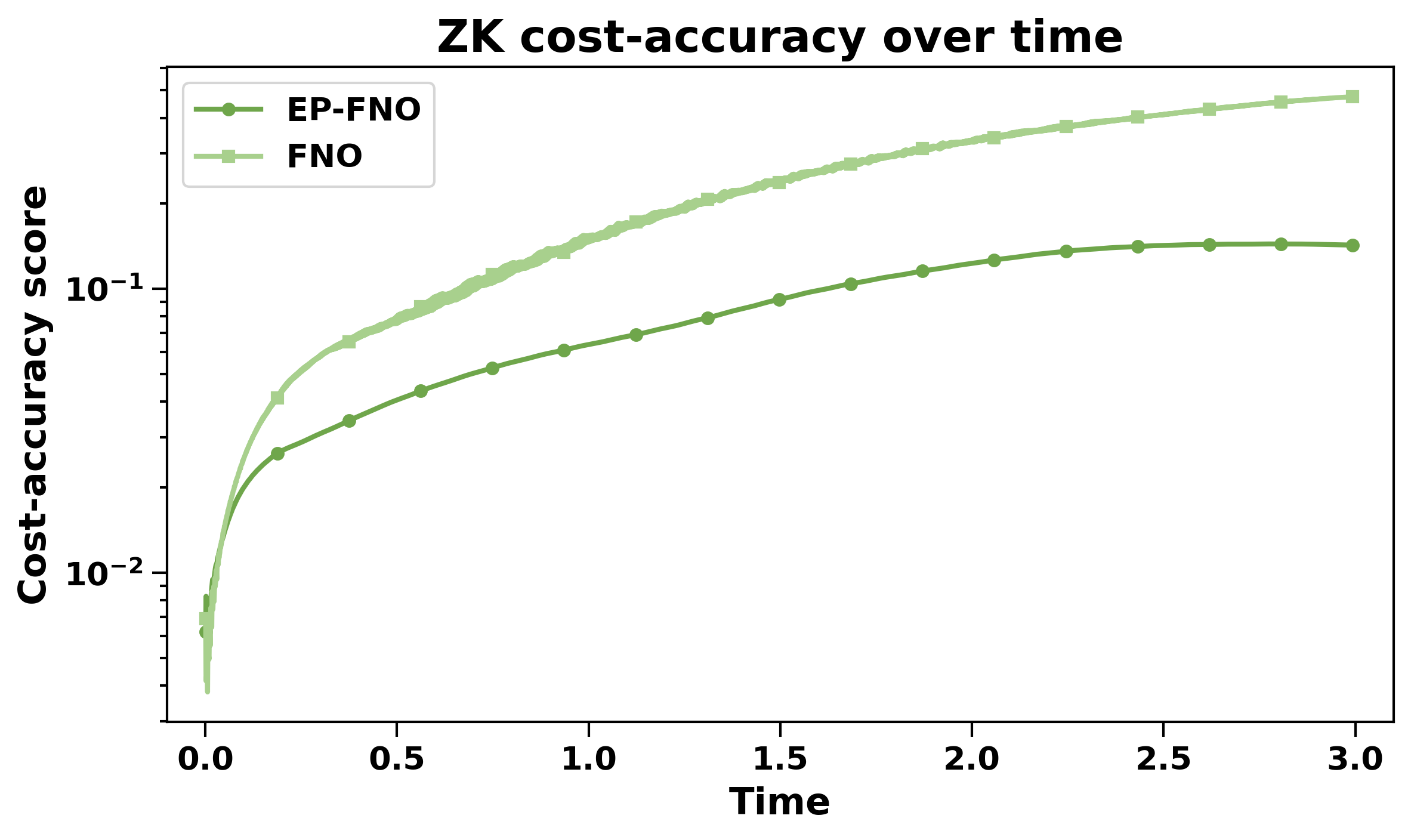}
\caption{ZK score over time.}
\label{fig:zk_cost_accuracy_time}
\end{subfigure}
\caption{Cost--accuracy comparison of FNO and EP-FNO. }
\label{fig:cost_accuracy_comparison}

\end{figure}

The cost--accuracy score is defined as the normalized cost multiplied by the
relative $L^2$ error. Smaller values therefore indicate a better trade-off
between computational cost and accuracy. Despite its projection overhead,
EP-FNO achieves a better final-time cost--accuracy score, with a $\mathbf{3.26\times}$
improvement over FNO as shown in Fig.~\ref{fig:cost_accuracy_comparison}.
The time-dependent cost--accuracy curves show that this advantage is maintained
over the long-time rollout.

\section{Discussion and Conclusion}
\label{sec:discussion}

\paragraph{Invariant structure.}
A structure-informed operator provides a natural way to improve 
long-time prediction of parametric Hamiltonian PDEs by making operators that are aware of conserved quantities.  A future direction is to develop more structure-based neural layers \cite{woodward2023mori} and
combine learned models with numerical solvers for dispersive PDEs. 

\paragraph{Model flexibility.}
The projection idea is not restricted to FNOs. Since it is applied as a
post-processing correction, the same structure-preserving strategy can be added
to other neural operators, recurrent models, or AI-based surrogate models.
\section*{Availability of Data and Code}
Data and code will be available upon publication

\section*{Conflict of Interest}
The authors declare that there are no conflicts of interest.

\section*{Funding}
Not applicable. The author received no specific funding for this work.

\section*{Author Contributions}
The authors are solely responsible for the conceptualization, methodology, software implementation, numerical experiments, analysis, and writing of this manuscript.

\bibliographystyle{plain}
\bibliography{ep_fno}

@article{chen2011multi,
  title={The multi-symplectic Fourier pseudospectral method for solving two-dimensional Hamiltonian PDEs},
  author={Chen, Yaming and Song, Songhe and Zhu, Huajun},
  journal={Journal of Computational and Applied Mathematics},
  volume={236},
  number={6},
  pages={1354--1369},
  year={2011},
  publisher={Elsevier}
}

@article{li2020fourier,
  title={Fourier neural operator for parametric partial differential equations},
  author={Li, Zongyi and Kovachki, Nikola and Azizzadenesheli, Kamyar and Liu, Burigede and Bhattacharya, Kaushik and Stuart, Andrew and Anandkumar, Anima},
  journal={arXiv preprint arXiv:2010.08895},
  year={2020}
}

@article{obieke2026energy,
  title={Energy Conserving Data Driven Discretizations for Maxwells Equations},
  author={Obieke, Victory},
  journal={arXiv preprint arXiv:2601.01902},
  year={2026}
}

@article{johnson2012new,
  title={New exact solutions for the sine-Gordon equation in 2+ 1 dimensions},
  author={Johnson, Suarez and Suarez, P and Biswas, A},
  journal={Computational Mathematics and Mathematical Physics},
  volume={52},
  number={1},
  pages={98--104},
  year={2012},
  publisher={Springer}
}

@article{kovachki2023neural,
  title={Neural operator: Learning maps between function spaces with applications to pdes},
  author={Kovachki, Nikola and Li, Zongyi and Liu, Burigede and Azizzadenesheli, Kamyar and Bhattacharya, Kaushik and Stuart, Andrew and Anandkumar, Anima},
  journal={Journal of Machine Learning Research},
  volume={24},
  number={89},
  pages={1--97},
  year={2023}
}

@article{tran2021factorized,
  title={Factorized fourier neural operators},
  author={Tran, Alasdair and Mathews, Alexander and Xie, Lexing and Ong, Cheng Soon},
  journal={arXiv preprint arXiv:2111.13802},
  year={2021}
}

@article{li2023fourier,
  title={Fourier neural operator with learned deformations for pdes on general geometries},
  author={Li, Zongyi and Huang, Daniel Zhengyu and Liu, Burigede and Anandkumar, Anima},
  journal={Journal of Machine Learning Research},
  volume={24},
  number={388},
  pages={1--26},
  year={2023}
}

@article{greydanus2019hamiltonian,
  title={Hamiltonian neural networks},
  author={Greydanus, Samuel and Dzamba, Misko and Yosinski, Jason},
  journal={Advances in neural information processing systems},
  volume={32},
  year={2019}
}

@article{jin2020sympnets,
  title={SympNets: Intrinsic structure-preserving symplectic networks for identifying Hamiltonian systems},
  author={Jin, Pengzhan and Zhang, Zhen and Zhu, Aiqing and Tang, Yifa and Karniadakis, George Em},
  journal={Neural Networks},
  volume={132},
  pages={166--179},
  year={2020},
  publisher={Elsevier}
}

@article{raissi2019physics,
  title={Physics-informed neural networks: A deep learning framework for solving forward and inverse problems involving nonlinear partial differential equations},
  author={Raissi, Maziar and Perdikaris, Paris and Karniadakis, George E},
  journal={Journal of Computational physics},
  volume={378},
  pages={686--707},
  year={2019},
  publisher={Elsevier}
}

@article{kissas2020machine,
  title={Machine learning in cardiovascular flows modeling: Predicting arterial blood pressure from non-invasive 4D flow MRI data using physics-informed neural networks},
  author={Kissas, Georgios and Yang, Yibo and Hwuang, Eileen and Witschey, Walter R and Detre, John A and Perdikaris, Paris},
  journal={Computer methods in applied mechanics and engineering},
  volume={358},
  pages={112623},
  year={2020},
  publisher={Elsevier}
}

@article{carlberg2013gnat,
  title={The GNAT method for nonlinear model reduction: effective implementation and application to computational fluid dynamics and turbulent flows},
  author={Carlberg, Kevin and Farhat, Charbel and Cortial, Julien and Amsallem, David},
  journal={Journal of Computational Physics},
  volume={242},
  pages={623--647},
  year={2013},
  publisher={Elsevier}
}

@article{chen2001multi,
  title={Multi-symplectic Fourier pseudospectral method for the nonlinear Schr{\"o}dinger equation},
  author={Chen, Jing-Bo and Qin, Meng-Zhao},
  journal={Electron. Trans. Numer. Anal},
  volume={12},
  pages={193--204},
  year={2001}
}

@article{burby2021normal,
  title={Normal stability of slow manifolds in nearly periodic Hamiltonian systems},
  author={Burby, Joshua William and Hirvijoki, Eero},
  journal={Journal of Mathematical Physics},
  volume={62},
  number={9},
  year={2021},
  publisher={AIP Publishing}
}

@misc{dipersio2025porthamiltonian,
      title={Port-Hamiltonian Neural Networks: From Theory to Simulation of Interconnected Stochastic Systems}, 
      author={Luca Di Persio and Matthias Ehrhardt and Youness Outaleb and Sofia Rizzotto},
      year={2025},
      eprint={2509.06674},
      archivePrefix={arXiv},
      primaryClass={math-ph},
      url={https://arxiv.org/abs/2509.06674}, 
}

@misc{zhong2024symp,
      title={Symplectic ODE-Net: Learning Hamiltonian Dynamics with Control}, 
      author={Yaofeng Desmond Zhong and Biswadip Dey and Amit Chakraborty},
      year={2024},
      eprint={1909.12077},
      archivePrefix={arXiv},
      primaryClass={cs.LG},
      url={https://arxiv.org/abs/1909.12077}, 
}

@article{najera2023structure,
  title={A structure-preserving neural differential operator with embedded Hamiltonian constraints for modeling structural dynamics},
  author={Najera-Flores, David A and Todd, Michael D},
  journal={Computational Mechanics},
  volume={72},
  number={2},
  pages={241--252},
  year={2023},
  publisher={Springer}
}

@misc{cranmer2020lagrangian,
      title={Lagrangian Neural Networks}, 
      author={Miles Cranmer and Sam Greydanus and Stephan Hoyer and Peter Battaglia and David Spergel and Shirley Ho},
      year={2020},
      eprint={2003.04630},
      archivePrefix={arXiv},
      primaryClass={cs.LG},
      url={https://arxiv.org/abs/2003.04630}, 
}

@article{woodward2023mori,
  title={Mori-Zwanzig mode decomposition: Comparison with time-delay embeddings},
  author={Woodward, Michael and Lin, Yen Ting and Tian, Yifeng and Hader, Christoph and Fasel, Hermann and Livescu, Daniel},
  journal={arXiv preprint arXiv:2311.09524},
  year={2023}
}

@inproceedings{chen2022kam,
  title={Kam theory meets statistical learning theory: Hamiltonian neural networks with non-zero training loss},
  author={Chen, Yuhan and Matsubara, Takashi and Yaguchi, Takaharu},
  booktitle={Proceedings of the AAAI Conference on Artificial Intelligence},
  volume={36},
  pages={6322--6332},
  year={2022}
}

@article{zhang2022vanishing,
  title={Vanishing Hall conductance for commuting Hamiltonians},
  author={Zhang, Carolyn and Levin, Michael and Bachmann, Sven},
  journal={Physical Review B},
  volume={105},
  number={8},
  pages={L081103},
  year={2022},
  publisher={APS}
}

@article{bokil2018highFDTD,
  author  = {Bokil, Vrushali A. and Gibson, Nathan L.},
  title   = {Analysis of Spatial High-Order Finite Difference Methods for Maxwell's Equations in Dispersive Media},
  journal = {IMA Journal of Numerical Analysis},
  volume  = {32},
  number  = {3},
  pages   = {926--956},
  year    = {2012},
  doi     = {10.1093/imanum/drr001},
  note    = {First published online 2011}
}

@article{banks2009finiteElement,
  author  = {Banks, H. T. and Bokil, Vrushali A. and Gibson, Nathan L.},
  title   = {Analysis of Stability and Dispersion in a Finite Element Method for Debye and Lorentz Dispersive Media},
  journal = {Numerical Methods for Partial Differential Equations},
  volume  = {25},
  number  = {4},
  pages   = {885--917},
  year    = {2009},
  doi     = {10.1002/num.20379}
}

@book{monk2003finiteElement,
  author    = {Monk, Peter},
  title     = {Finite Element Methods for Maxwell's Equations},
  publisher = {Oxford University Press},
  address   = {Oxford},
  year      = {2003},
  series    = {Numerical Mathematics and Scientific Computation},
  isbn      = {9780198508885}
}

@book{trefethen2000spectral,
  author    = {Trefethen, Lloyd N.},
  title     = {Spectral Methods in MATLAB},
  publisher = {Society for Industrial and Applied Mathematics},
  address   = {Philadelphia, PA},
  year      = {2000},
  doi       = {10.1137/1.9780898719598}
}

@book{fornberg1996pseudospectral,
  author    = {Fornberg, Bengt},
  title     = {A Practical Guide to Pseudospectral Methods},
  publisher = {Cambridge University Press},
  address   = {Cambridge},
  year      = {1996},
  series    = {Cambridge Monographs on Applied and Computational Mathematics},
  doi       = {10.1017/CBO9780511626357}
}

@article{bridges2001multisymplecticSpectral,
  author  = {Bridges, Thomas J. and Reich, Sebastian},
  title   = {Multi-Symplectic Spectral Discretizations for the Zakharov--Kuznetsov and Shallow Water Equations},
  journal = {Physica D: Nonlinear Phenomena},
  volume  = {152--153},
  pages   = {491--504},
  year    = {2001},
  doi     = {10.1016/S0167-2789(01)00188-9}
}

@article{oguadimma2026foundational,
  author  = {Oguadimma, Emmanuel E. and Elbarkawy, Mohamed A. F. and Oranugo, Dominic O. and Salem, Heba E. and Bayram, Mustafa and Obulezi, Okechukwu J.},
  title   = {A Foundational Review of Ordinary Differential Equation Solution Methods and Their Inherent Symmetries},
  journal = {Boletim da Sociedade Paranaense de Matem{\'a}tica},
  volume  = {44},
  number  = {8},
  pages   = {1--27},
  year    = {2026},
  doi     = {10.5269/bspm.80626}
}

@article{Yee1966,
  author  = {Yee, Kane S.},
  title   = {Numerical Solution of Initial Boundary Value Problems Involving Maxwell's Equations in Isotropic Media},
  journal = {IEEE Transactions on Antennas and Propagation},
  volume  = {14},
  number  = {3},
  pages   = {302--307},
  year    = {1966},
  doi     = {10.1109/TAP.1966.1138693}
}

@book{whitham1974linear,
  author    = {Whitham, G. B.},
  title     = {Linear and Nonlinear Waves},
  publisher = {Wiley-Interscience},
  address   = {New York},
  year      = {1974}
}

@book{drazin1989solitons,
  author    = {Drazin, P. G. and Johnson, R. S.},
  title     = {Solitons: An Introduction},
  series    = {Cambridge Texts in Applied Mathematics},
  publisher = {Cambridge University Press},
  address   = {Cambridge},
  year      = {1989}
}

@book{sulem1999nonlinear,
  author    = {Sulem, Catherine and Sulem, Pierre-Louis},
  title     = {The Nonlinear Schr{\"o}dinger Equation: Self-Focusing and Wave Collapse},
  series    = {Applied Mathematical Sciences},
  volume    = {139},
  publisher = {Springer},
  address   = {New York},
  year      = {1999},
  doi       = {10.1007/b98958}
}

@article{morrison2005hamiltonian,
  author  = {Morrison, P. J.},
  title   = {Hamiltonian and action principle formulations of plasma physics},
  journal = {Physics of Plasmas},
  volume  = {12},
  number  = {5},
  pages   = {058102},
  year    = {2005},
  doi     = {10.1063/1.1882353}
}

@article{morrison1980noncanonical,
  author  = {Morrison, P. J. and Greene, J. M.},
  title   = {Noncanonical Hamiltonian Density Formulation of Hydrodynamics and Ideal Magnetohydrodynamics},
  journal = {Physical Review Letters},
  volume  = {45},
  number  = {10},
  pages   = {790--794},
  year    = {1980},
  doi     = {10.1103/PhysRevLett.45.790}
}

@article{morrison1998hamiltonian,
  author  = {Morrison, P. J.},
  title   = {Hamiltonian description of the ideal fluid},
  journal = {Reviews of Modern Physics},
  volume  = {70},
  number  = {2},
  pages   = {467--521},
  year    = {1998},
  doi     = {10.1103/RevModPhys.70.467}
}

@book{boyd2020nonlinear,
  author    = {Boyd, Robert W.},
  title     = {Nonlinear Optics},
  edition   = {4},
  publisher = {Academic Press},
  year      = {2020},
  isbn      = {978-0-12-811002-7}
}

@book{agrawal2019nonlinear,
  author    = {Agrawal, Govind P.},
  title     = {Nonlinear Fiber Optics},
  edition   = {6},
  publisher = {Academic Press},
  year      = {2019},
  isbn      = {978-0-12-817042-7}
}

@misc{dauner2024residual,
  title        = {Residual Factorized Fourier Neural Operator for simulation of three-dimensional turbulence},
  author       = {Dauner, Maximilian and Bheemanakone, Ankith and Z{\"o}nnchen, Benedikt and Socher, Gudrun},
  year         = {2024},
  note         = {Submitted to ICLR 2024},
  url          = {https://openreview.net/forum?id=yGdoTL9g18}
}

@article{valente2025physics,
  title   = {Physics-consistent machine learning with output projection onto physical manifolds},
  author  = {Valente, Matilde and Dias, Tiago C. and Guerra, Vasco and Ventura, Rodrigo},
  journal = {Communications Physics},
  volume  = {8},
  number  = {433},
  year    = {2025},
  doi     = {10.1038/s42005-025-02329-1}
}

@misc{liu2025adaptive,
  title         = {Adaptive Correction for Ensuring Conservation Laws in Neural Operators},
  author        = {Liu, Chaoyu and Li, Yangming and Deng, Zhongying and Budd, Chris and Sch{\"o}nlieb, Carola-Bibiane},
  year          = {2025},
  eprint        = {2505.24579},
  archivePrefix = {arXiv},
  primaryClass  = {cs.LG},
  doi           = {10.48550/arXiv.2505.24579}
}

@inproceedings{duruisseaux2024hard,
  title     = {Towards Enforcing Hard Physics Constraints in Operator Learning Frameworks},
  author    = {Duruisseaux, Valentin and Liu-Schiaffini, Miguel and Berner, Julius and Anandkumar, Anima},
  booktitle = {ICML 2024 Workshop on AI for Science},
  year      = {2024},
  url       = {https://openreview.net/forum?id=Zvxm14Rd1F}
}

@article{cardosobihlo2025exactly,
  title   = {Exactly conservative physics-informed neural networks and deep operator networks for dynamical systems},
  author  = {Cardoso-Bihlo, Elsa and Bihlo, Alex},
  journal = {Neural Networks},
  volume  = {181},
  pages   = {106826},
  year    = {2025},
  doi     = {10.1016/j.neunet.2024.106826}
}

@article{kovachki2021universal,
  title={On universal approximation and error bounds for Fourier neural operators},
  author={Kovachki, Nikola and Lanthaler, Samuel and Mishra, Siddhartha},
  journal={Journal of Machine Learning Research},
  volume={22},
  number={290},
  pages={1--76},
  year={2021}
}

@article{obieke2025structure,
  title={Structure-Preserving Physics-Informed Neural Network for the Korteweg--de Vries (KdV) Equation},
  author={Obieke, Victory and Oguadimma, Emmanuel},
  journal={arXiv preprint arXiv:2511.00418},
  year={2025}
}

@article{gibson2026analysis,
  title={Analysis of Nonlinear Random Polarization in Dispersive Dielectrics},
  author={Gibson, Nathan L and Oguadimma, Emmanuel E},
  journal={arXiv preprint arXiv:2606.27445},
  year={2026}
}

\appendix

\renewcommand{\thesubsection}{\Alph{section}.\arabic{subsection}}

\section{Supplementary Results}

\subsection{Zakharov--Kuznetsov Equation}\label{appenZK}
\begin{figure}[htpb]
    \centering
    \includegraphics[width=0.7\linewidth]{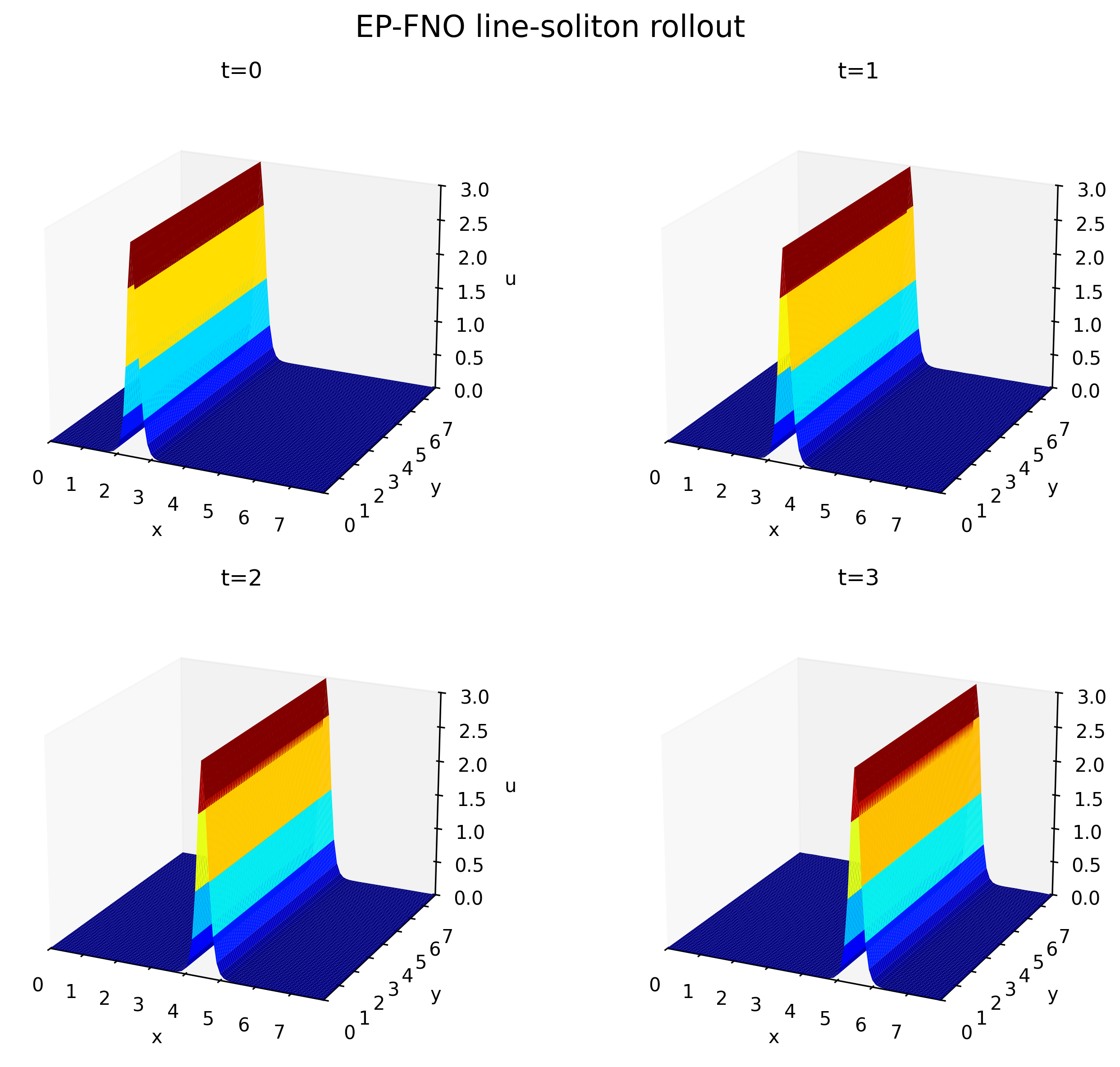}
        \caption{ZK line Experiment}
        \label{fig:zk_line}
\end{figure}

\begin{figure}[htbp]
    \centering
    \begin{subfigure}{0.48\textwidth}
        \centering
        \includegraphics[width=\linewidth]{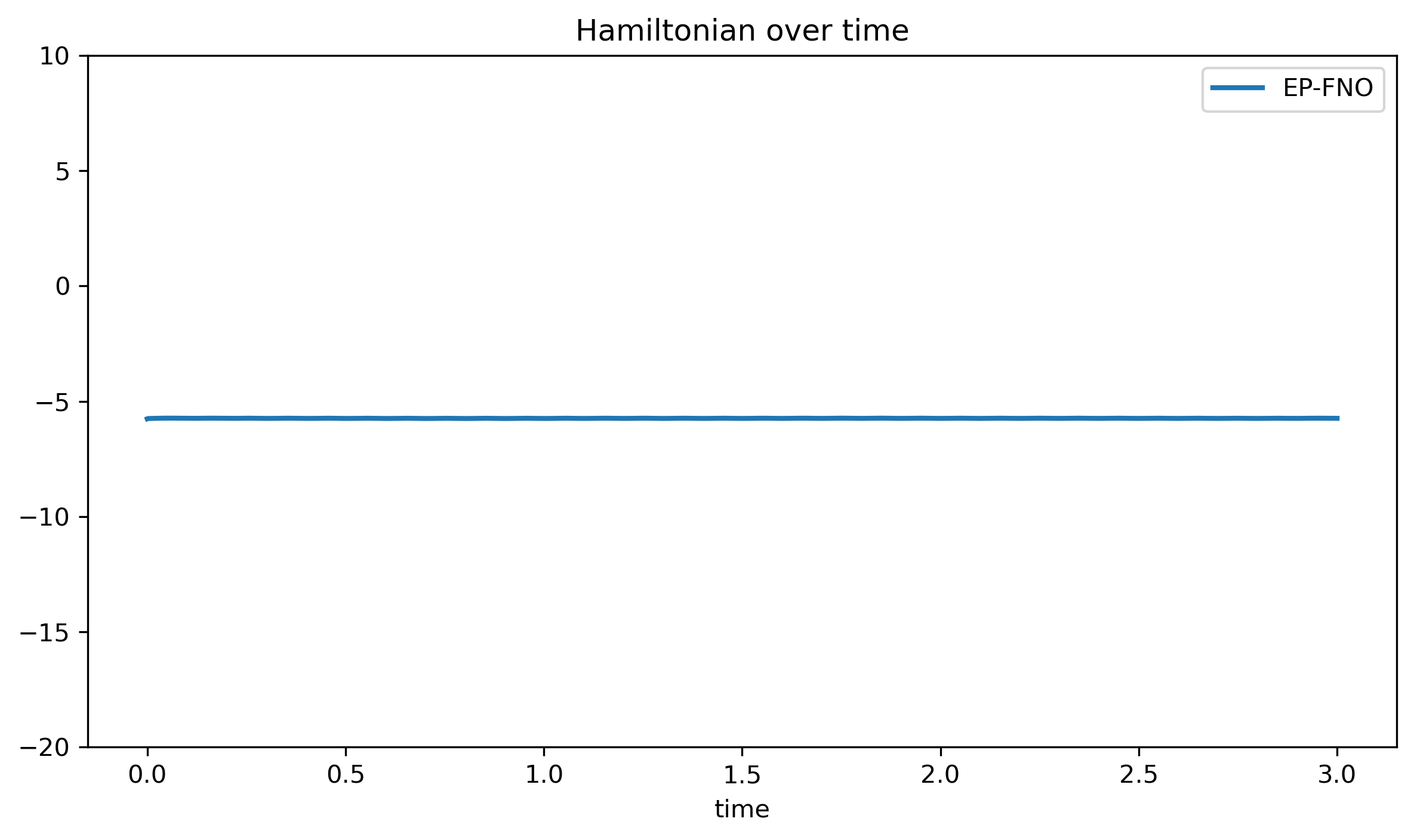}
        \caption{}
        \label{fig:zklinelh}
    \end{subfigure}
    \hfill
     \begin{subfigure}{0.48\textwidth}
        \centering
        \includegraphics[width=\linewidth]{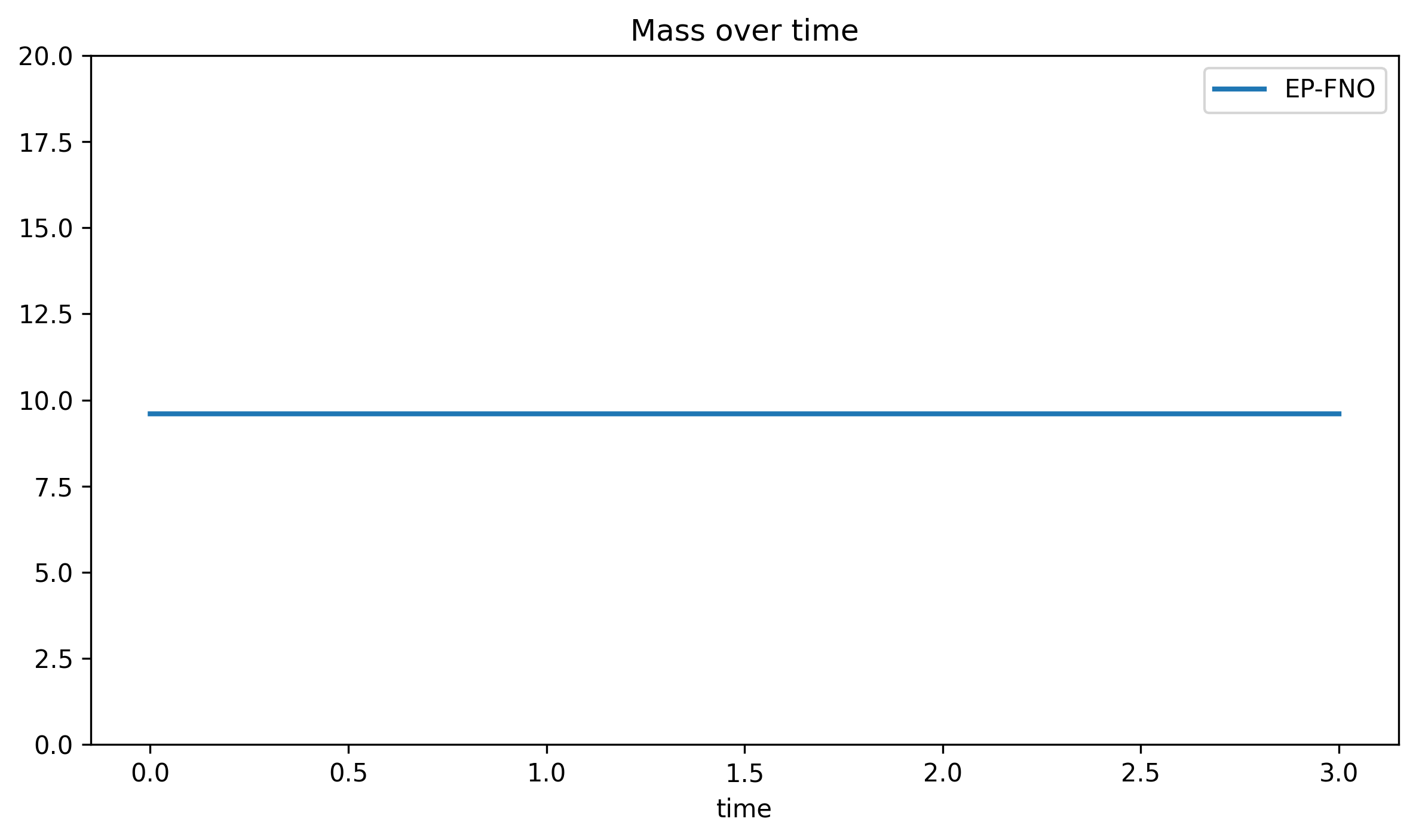}
        \caption{}
        \label{fig:zklinecycl}
    \end{subfigure}

    \caption{ Hamiltonian and mass of the ZK line Experiment EP-FNO}
    \label{fig:zklineHandm}
\end{figure}

\subsection{FNO Results}
\label{appenKp}

To further highlight the motivation for the proposed EP-FNO framework, we report the corresponding FNO results for the parametric KP equation. Among the models considered, the KP equation is particularly challenging because it combines nonlinear wave propagation, higher-order dispersion, and anisotropic two-dimensional dynamics. Small phase and speed errors can accumulate rapidly during long autoregressive rollouts, especially for line-soliton solutions. This makes the KP equation a useful test case for assessing the limitations of the standard FNO and the advantages of the energy-preserving approach.
\begin{figure}[htpb]
    \centering
    \includegraphics[width=0.7\linewidth]{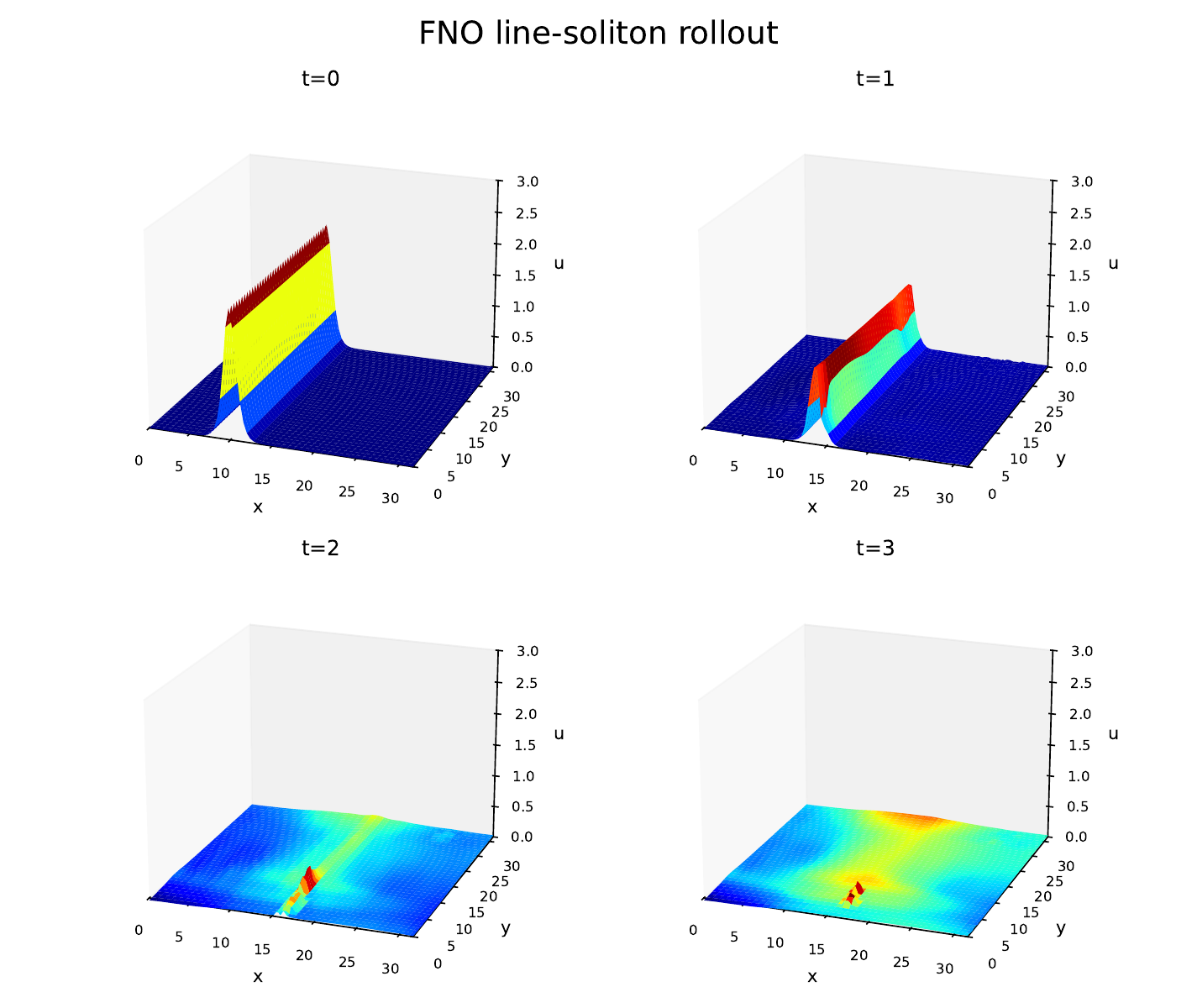}
        \caption{FNO rollout for the KP line-soliton benchmark at representative times \(t=0,1,2,\) and \(3\).}
        \label{fig:kp_fno_surface_result1}
\end{figure}

Figure~\ref{fig:kp_fno_surface_result1} shows the corresponding FNO rollout for the KP line-soliton benchmark. In contrast to EP-FNO, the predicted solution progressively departs from the characteristic line-soliton profile during autoregressive prediction. By (t=2) and (t=3), substantial distortion of the wave structure is observed.

\begin{figure}[H]
    \centering
    \begin{subfigure}{0.48\textwidth}
        \centering
        \includegraphics[width=\linewidth]{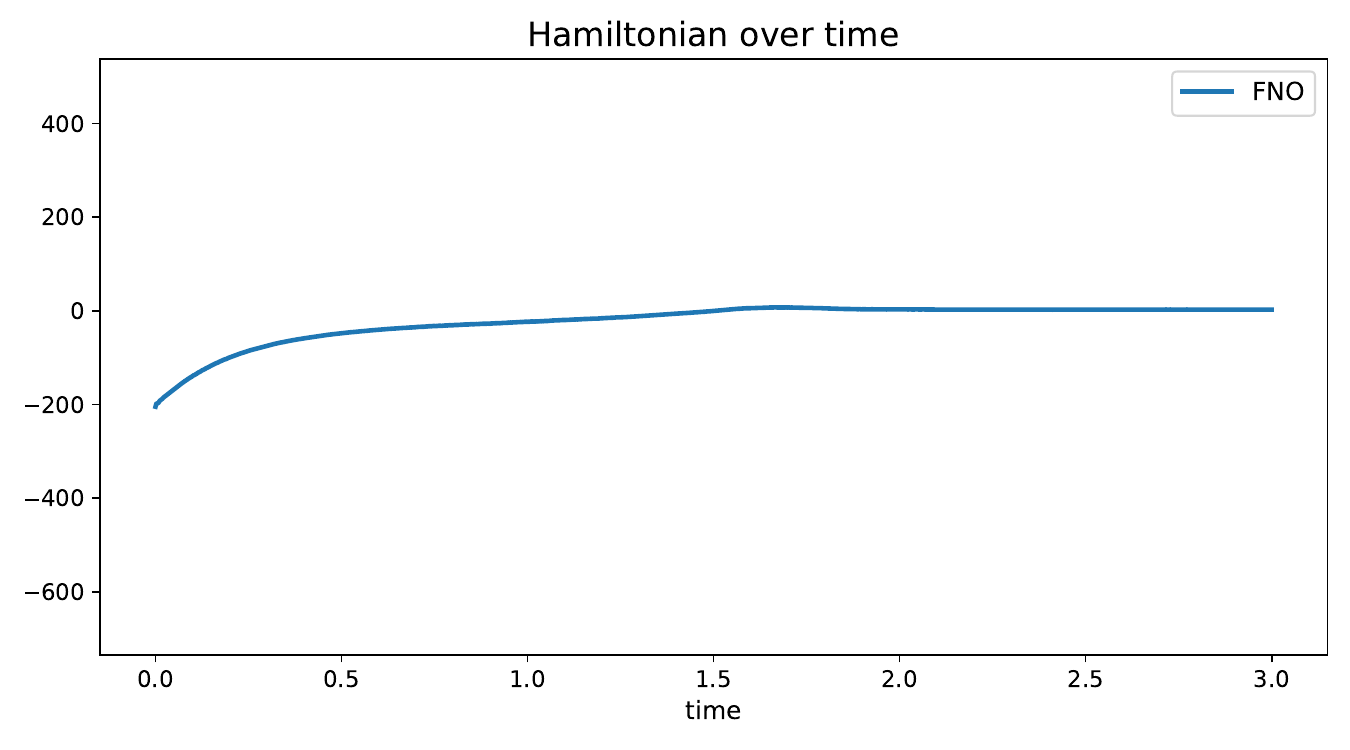}
        \caption{}
        \label{fig:kplinelh}
    \end{subfigure}
    \hfill
     \begin{subfigure}{0.48\textwidth}
        \centering
        \includegraphics[width=\linewidth]{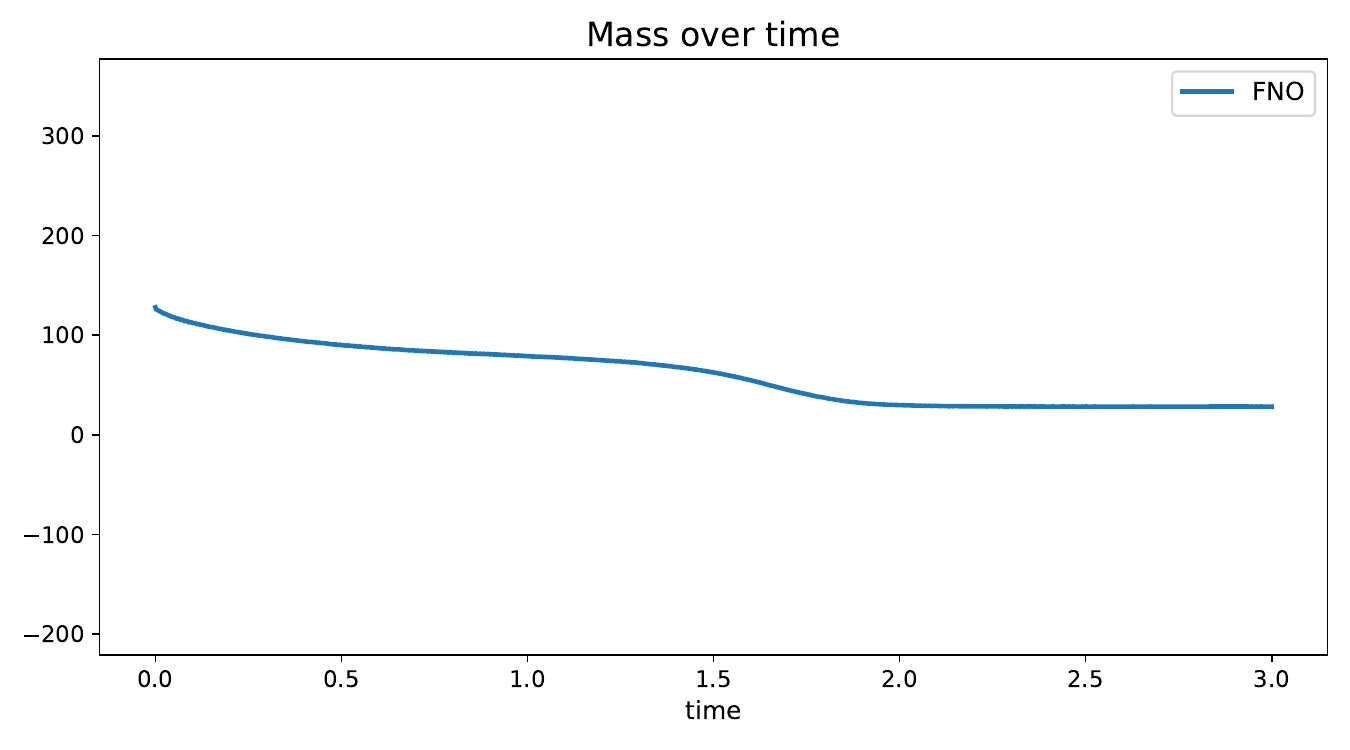}
        \caption{}
        \label{fig:kplinem}
    \end{subfigure}

    \caption{Evolution of the Hamiltonian and mass during the FNO rollout for the KP line-soliton benchmark.}
    \label{fig:kpfnoHandm}
\end{figure}

Figure~\ref{fig:kpfnoHandm} shows the evolution of the Hamiltonian and mass during the same rollout. Both quantities exhibit significant drift from their initial values, indicating that the standard FNO does not preserve the underlying invariants of the KP equation. The loss of invariant preservation is consistent with the degradation observed in the predicted solution and provides additional evidence for the benefits of the structure-preserving projection employed by EP-FNO.

\end{document}